\documentclass[a4paper,10pt, english]{amsart}
\usepackage[utf8]{inputenc}
\usepackage[T1]{fontenc}
\usepackage{lmodern}
\usepackage{babel}
\usepackage{amsmath}
\usepackage{amssymb}
\usepackage{amsfonts}
\usepackage{amsthm}
\usepackage{mathrsfs}
\usepackage{pgf}
\usepackage{tikz}
\usepackage{csquotes}
\usepackage{enumerate}
\usepackage[top=2.5cm, bottom=2.5cm, left=2.5cm, right=2.5cm]{geometry}
\usepackage{hyperref}
\numberwithin{equation}{section}

\usetikzlibrary{patterns.meta} 
\theoremstyle{plain}
\newtheorem{prop}{Proposition}[section]
\newtheorem{lemma}[prop]{Lemma}
\newtheorem{theorem}[prop]{Theorem}
\newtheorem{coro}[prop]{Corollary}

\newtheorem{property}[prop]{Property}
\theoremstyle{definition}
\newtheorem{defi}[prop]{Definition}

\newcommand{\N}{\mathbb{N}} 
\newcommand{\C}{\mathbb{C}} 
\newcommand{\Z}{\mathbb{Z}} 
\newcommand{\R}{\mathbb{R}} 

\newcommand{\tr}{\mathrm{tr}}

\newcommand{\iu}{{\rm i}}
\newcommand{\e}{{\rm e}}
\renewcommand{\d}{{\rm d}}

\def\ec{{\mathbb E}}
\def\pc{{\mathbb P}}
\newcommand{\dist}{\mathrm{dist}}
\newcommand{\supp}{\mathrm{supp}}
\newcommand{\eb}{\overline{\mathbb{E}}}

\newcommand{\Hilb}{\mathcal{H}}

\newcommand{\Dom}{\mathrm{Dom}}

\renewcommand{\Re}{\mathrm{Re}}
\renewcommand{\Im}{\mathrm{Im}}

\renewcommand{\leq}{\leqslant}
 
 \renewcommand{\geq}{\geqslant} 
\title{Localization for one-dimensional Anderson-Dirac models}
\author[S. Zalczer]{Sylvain Zalczer}
\address{BCAM – Basque Center for Applied Mathematics, Mazarredo 14, E48009
Bilbao, Basque Country – Spain.}
\email{szalczer@bcamath.org}
\begin{document}

\begin{abstract}
We prove spectral and dynamical localization for a one-dimensional Dirac operator to which is added an ergodic random potential, with a discussion on the different types of potential. We use scattering properties to prove the positivity of the Lyapunov exponent through Fürstenberg theorem. We get then the Hölder regularity of the integrated density of states through a new version of Thouless formula, and thus the Wegner estimate necessary for the multiscale analysis.
\end{abstract}
\maketitle

\section{Introduction}

The goal of this paper is to get localization for random Dirac operators in dimension 1, i.e. operators of the form
\begin{equation*}
 H_\omega=D_0+V_{per}+V_\omega\text{ on }L^2(\R,\C^2).
\end{equation*}
In this equation, $D_0=J\frac{\d}{\d x}$, with $J=\begin{pmatrix}0&-1\\1&0\end{pmatrix}$, is the free Dirac operator.
The potential $V_{per}$ is supposed to be 1-periodic and the random potential $V_\omega$ is an Anderson ergodic potential:
\[V_\omega(x)=\sum_{n\in\Z}\lambda_n(\omega)u(x-n)\]
with random coupling constants $\lambda_n(\omega)$ and an elementary potential $u$ whose characteristics will be precised. As it is well-known, the Dirac operator has been defined to study relativistic matter~\cite{thaller}. On the other hand, the 2-dimensional Dirac operator can be used to model graphene~\cite{CGPNGG}. It could thus be  relevant to use, at least as a first step, a 1-dimensional Dirac operator to model a graphene nanoribbon.


Since the first works of Anderson~\cite{Anderson}, random operators have been given a lot of interest. One-dimensional models, in particular, have been extensively studied~\cite{GMP77,CKM,DSS}. Indeed,  specific methods -- involving Lyapunov exponents -- have been developed for this setting. It has thus been possible to prove localization on the whole spectrum for these models, while in higher dimension results are known only for energies near band edges.

Most of these works are about random \emph{Schrödinger} operators, i.e. with a kinetic energy given by a (discrete or continuous) Laplacian. Only few papers have addressed the question of localization for random \emph{Dirac} operators. In~\cite{PdO05}, Prado and de Oliveira consider a discrete Dirac operator with a random potential taking only 
 2 values. They prove, using a multiscale analysis, dynamical localization, on the whole spectrum for positive mass and outside a discrete critical set of energies for zero mass. In \cite{PdO07}, they give a deeper analysis of delocalization with the same kind of randomness, including the case of operators on the continuum. In~\cite{PdOC}, together with Carvalho, they prove -- among others -- dynamical localization on the whole spectrum for a discrete model, in the case where the random potential has an absolutely continuous distribution. They use here a different method, called fractional moments method. In~\cite{BCZ}, the author, together with Barbaroux and Cornean, proved localization near band edges in any dimension, in the case of absolutely continuous random variables. This paper has been followed by~\cite{Zdos}, where Lipschitz continuity of the integrated density of states is proven.
 
 The organization of the paper is the following. In section~\ref{sec:2}, we present in detail the model and the results. In section~\ref{sec:3}, we prove the positivity of the Lyapunov exponent of the group of transfer matrices, using properties of local uniqueness of Weyl-Titchmarsh functions. In section~\ref{sec:4}, we get Hölder regularity for this exponent, using Grönwall-type estimates, and, in section~\ref{sec:5}, we obtain the same result for the integrated density of states. This last point needs to prove a new Thouless formula, adapted to Dirac operators. Finally, in section~\ref{sec:6}, we check all the hypotheses of the bootstrap multiscale analysis proposed by Germinet and Klein in~\cite{GKbootstrap}, which concludes the proof of our main result.  In appendix~\ref{app:A}, we explain in detail the uniqueness properties for Weyl-Titchmarsh functions and, in appendix~\ref{app:B}, we prove that the density of states can be defined using Dirichlet boundary conditions.

 \subsection*{Aknowledgements}
 The author is grateful to H. Boumaza for giving him a short introduction to Lyapunov exponent methods and for a reading.
 This research is supported by the Basque Government
through the BERC 2018-2021 program and by the Ministry of
Science, Innovation and Universities: BCAM Severo Ochoa
accreditation SEV-2017-0718 and the mineco project PID2020-112948GB-I00.
\section{Model and results}\label{sec:2}
As mentioned in the introduction, the operator studied in this paper will be
\begin{equation}
 H_\omega=D_0+V_{per}+V_\omega\text{ on }L^2(\R,\C^2).
\end{equation}
We recall that the free Dirac operator is $D_0=J\frac{\d}{\d x}$, with $J=\begin{pmatrix}0&-1\\1&0\end{pmatrix}$.
We assume that $V_{per}$ is a 1-periodic bounded function  with values in the space of 2-by-2 real symmetric matrices.

In the context of Dirac operators, it is relevant to decompose the potentials according to the Pauli matrices:
\[\sigma_1=\begin{pmatrix}0&1\\1&0\end{pmatrix},\ \sigma_2=\begin{pmatrix}0&-\iu\\\iu&0\end{pmatrix},\ \sigma_3=\begin{pmatrix}1&0\\0&-1\end{pmatrix}.\] We write 
\[V_{per}=v_{per}^{am}\sigma_1+v_{per}^{sc}\sigma_3+v_{per}^{el}I_2,\]
where $v_{per}^{am}$ is called \emph{anomalous magnetic moment},  $v_{per}^{sc} $ is called  \emph{scalar potential} and $v_{per}^{el}$ is called \emph{electrostatic potential}. Each of these functions is real-valued, measurable and bounded.
We do not include a  potential  multiple of $\sigma_2$, which would be called \emph{magnetic moment}, since it can easily be proven that, if we denote by $D=D_0+V$ a Dirac operator with a matrix-valued potential, $D+v_{mg}(x)\sigma_2$ is unitarily equivalent to $D$ (cf.~\cite{BEKT}).

The random potential $V_\omega$ is written as:
\[V_\omega(x)=\sum_{n\in\Z}\lambda_n(\omega)u(x-n).\]
The elementary potential $u$ is supposed to have support in $[-1/2,1/2]$. It can be decomposed just as the periodic potential. Nevertheless, it is not possible to prove localization for all configurations. In particular, for any potential $V_{el}\in L^1_{loc}(\R)$, the operator $D_0+V_{el}I_2$ is unitarily equivalent to $D_0$ (this is the well-known Klein's paradox). We will consider the two following cases:
\begin{enumerate}
 \item $\label{potelem}u(x)=q_{am}(x)\sigma_1+q_{sc}(x)\sigma_3,$
 \item $u(x)=q_{el}(x)I_2$ and $\supp~q_{el}\cap \supp~[ (v_{per}^{am})^2+(v_{per}^{sc})^2]$ has positive measure.
\end{enumerate}
In both cases, the functions  $q_{am}$, $q_{sc}$ and $q_{el}$ are measurable, bounded and  real-valued. The way we deal with these different potentials is explained in Appendix~\ref{app:A}.
 
 The coupling contants $\lambda_n$ are independent, identically distributed real random variables on a complete probability space $\Omega$. We assume that their common distribution is bounded and non-trivial, in the sense that their support contains at least two points.  
 Note that, under these conditions, the operator has real coefficients.
With standard methods, we see that $H_\omega$ is self-adjoint, with domain $H^1(\R,\C^2)$.

Under these conditions, the family $(H_\omega)_{\omega\in\Omega}$ is an ergodic random operator. As it is known from Pastur's theorem (see for example~\cite[Theorem~4.3]{kirschbook}), the spectrum of $H_\omega$ is almost surely constant:
 there exists $\Sigma\subset \R$ such that $\sigma(H_\omega)=\Sigma$ with probability 1.
 Similarly, the pure point, singular continuous and absolutely continuous parts of the spectrum are almost surely constant.

We want to prove Anderson localization for this operator.
Our first result will be the following.
\begin{theorem}[Exponential localization]\label{thm:loc-spec}
 With probability 1, the operator $H_\omega$ has pure point spectrum, in the sense that it has no continuous spectrum. Moreover, all its eigenfunctions are  exponentially decaying at $\pm\infty$.
\end{theorem}

We will be able to prove a stronger property, called \emph{dynamical localization}.
\begin{theorem}[Dynamical localization]\label{thm:loc-dyna}
There exists a discret set $M\subset\R$ such that, for every compact interval $I\subset\R\backslash M$, every compact set $K\subset \R$, and every $p>0$,
\begin{equation}
 \ec\left(\sup_{t>0}\||X|^p\e^{-\iu t H_\omega}P_I(H_\omega)\chi_K\|\right)<+\infty,
\end{equation}
where $P_I$ is the spectral projection on $I$.
\end{theorem}
The existence of a set $M$ of critical energies for which it exists (dynamically) delocalized states is a well-known fact in one-dimensional disordered systems (cf. \cite{DSS} and references therein). In the case of Dirac operators, quasi-ballistic behaviour at critical energies has been studied in~\cite{PdO07}.
The critical set does not appear in Theorem~\ref{thm:loc-spec} since it is discrete, and a discrete set cannot support continuous spectrum.

We will follow the classic methods of proof of localization for one-dimensional systems, as they are  used for example in~\cite{DSS} for Schrödinger operators. The fundamental fact is that  a non-zero potential induces a nontrivial scattering. Henceforth, we will get the hypotheses of  Fürstenberg's theorem which will give us the positivity of the Lyapunov exponent of the group of transfer matrices, which is called  Fürstenberg's group. With the help of Grönwall-type estimates, we will be able to get the Hölder continuity of this exponent.

The next step is to get the same regularity on the integrated density of states. In \cite{Zdos}, we proved that it was possible the define the density of states measure for random Dirac operators by, given a bounded measurable function $\phi$ on $\R$,
\begin{equation}\label{def-dos}\nu(\phi)=\lim_{L\to\infty}\frac{1}{L}\tr(\chi_L\phi(H_\omega)\chi_L),\end{equation}where $\chi_L$ is the characteristic function of $[-L/2,L/2]$ (the result exists in any dimension, but we give here only the one-dimensional version). 
By extension, given a Borel set $B$, we define the density of states measure by $\nu(B)=\nu(\chi_B)$, where $\chi_B$ is the characteristic function of $B$.
Since Dirac operators are not bounded from below, we cannot define the integrated density of states in the usual way as $N(E)=\nu((-\infty,E])$. Nevertheless, we prove a result of regularity.
\begin{theorem}\label{thm:reguDOS}
Let $M$ be the discrete set already mentioned in Theorem~\ref{thm:loc-dyna} and fix a compact interval $I\subset\R\backslash M$.
 There exist $\alpha>0$ and $C<\infty$ such that, for all $E<E'\in I$,
 \begin{equation}
  \nu([E,E'])\leq C(E'-E)^\alpha.
 \end{equation}
\end{theorem}

The relation between the Lyapunov exponent and the density of states is done through a result called Thouless formula. Note that, because of the asymptotic behaviour of the integrated density of states, we need to prove a new Thouless formula adapted to Dirac operators.

While in~\cite{Zdos} we get such a regularity in a spectral gap of the unperturbed operator from a Wegner estimate previously proven in~\cite{BCZ} to get Anderson localization, here, we use the regularity of the density of states to get the Wegner estimate and, then, Anderson localization on the whole real line. Our last section is devoted to the proof of the conditions to apply the multiscale analysis published by Germinet and Klein in~\cite{GKbootstrap}.

\section{Positivity of the Lyapunov exponent}\label{sec:3}
\subsection{Floquet theory}
Our preliminary step will be to study Floquet theory for the periodic operator $H_0=D_0+V_{\text{per}}$. We define some remarkable elements which will be used in the following, as they are introduced for example in the first two chapters of \cite{Eastham}. In that book, the considered system is a second-order scalar differential equation. Although our setting is different, we can work similarly since the quantity of interest in the second-order case is the vector $\begin{pmatrix}u\\u'\end{pmatrix}$. In our case, the function $u$ is already a vector. As a consequence, we will be able to use similar techniques. In particular, by analogy, we will call \emph{Dirichlet boundary conditions} on an interval $(a,b)$ (for a solution $\begin{pmatrix}                                                                                                                                                                                                                                                                                                                                                                                                                                                                                                                                                                                                                                                        u^\uparrow\\u^\downarrow                                                                                                                                                                                                                                            \end{pmatrix}$) the condition $u^\uparrow(a)=u^\uparrow(b)=0$.

For $z\in\C$, we consider the equation 
\begin{equation}\label{floquet}
 D_0u+V_{\text{per}}u=zu.
\end{equation}
By analogy with the second-order case, we denote by $u_N(\cdot,z)$ and $u_D(\cdot,z)$ the solutions to~\eqref{floquet} satisfying $u_N(-1/2,z)=\begin{pmatrix}1\\0\end{pmatrix}$ and $u_D(-1/2,z)=\begin{pmatrix}0\\1\end{pmatrix}$.

We denote by $g_0(z)$ the transfer matrix between -1/2 and 1/2, defined as the 2-by-2 matrix of which the columns are $u_N(1/2,z)$ and $u_D(1/2,z)$. We see that this matrix is entire in $z$. Its eigenvalues satisfy the characteristic equation \begin{equation}\label{char-eq}\rho^2-D(z)\rho+1=0,\end{equation} where $D(z)=\tr g_0(z)$.
We can then define 2 algebraic functions $\rho_\pm$ (with singularities at points where $D(z)=\pm2$) as the roots of this equation. 

Let us consider a \emph{stability interval}, i.e. a maximal real interval $(E_-,E_+)$ such that $|D(E)|<2$ for all $E\in(E_-,E_+)$. Since the equation has real coefficients, both $u_N(\cdot,E)$ and $u_D(\cdot,E)$ are real-valued for $E\in\R$ and thus, for $E\in(E_-,E_+)$, one has that
\begin{equation}\label{defrho}\rho_\pm(E)=\frac12(D(E)\pm\iu\sqrt{4-D(E)^2}),\end{equation}
 with $|\rho_\pm(E)|=1$ and $\rho_-(E)=\overline{\rho_+(E)}$.
 
 Let us now consider the strip\[S=\{z=E+\iu\eta, E\in(E_-,E_+), \eta\in\R\}.\]
 According to Theorem~3.1.1 of \cite{Eastham}, the eigenvalues of the problem with Dirichlet boundary conditions are not in the stability intervals.
 Thus, for $z\in S$, we can define $v_\pm(z)$ as the eigenvectors of $g_0(z)$ corresponding to the eigenvalues $\rho_\pm(z)$ with the first component normalized to one:\begin{equation}\label{defc}v_\pm(z)=\begin{pmatrix}1\\c_\pm(z)\end{pmatrix}.\end{equation}
 An explicit calculation gives that \[c_\pm(z)=\frac{\rho_\pm(z)-u_N(1/2,z)}{u_D(1/2,z)},\] where we know that $u_D(1/2,z)\neq0$ in $S$ since the Dirichlet eigenvalues are real and outside the stability intervals. As a consequence, $c_\pm$ and then $v_\pm$ are analytic in $S$ with at most algebraic singularities at $E_-$ and $E_+$. 
  
  We will denote by $\phi_\pm$ the \emph{Floquet solutions}, i.e. the solutions satisfying \begin{equation}
\phi_\pm(-1/2,z)=v_\pm(z)~;
  \end{equation}
these solutions are such that 
\[\phi_\pm(x+1,z)=\rho_\pm(z)\phi_\pm(x,z)\] for all $x$ in $\R$. Since $\rho_+(z)\neq\rho_-(z)$ on $S$, these two functions are a fundamental system of solutions to~\eqref{floquet}.
Moreover, for fixed $x$, the function $\phi_\pm(x,\cdot)$ are both analytic in $S$, with at most algebraic singularities at $E_-$ and $E_+$.

Last, we see that, if $\eta>0$, $\phi_\pm(\cdot,E+\iu\eta)$ decays exponentially near $\pm\infty$ and if $\eta<0$, $\phi_\pm(\cdot,E+\iu\eta)$ decays exponentially near $\mp\infty$.
But we know, from the corollary given just after Theorem~6.8 in \cite{Weidmann}, that, for any potential, Dirac operator is in the so-called ``limit point case'' at $\pm\infty$, which means that, for any $z\in\C\backslash\R$, there exists a unique (up to multiplication) solution to $\eqref{floquet}$, called \emph{Weyl solution}, which is in $L^2$ in a neighbourhood of $\pm\infty$. Here, we see that the Floquet solutions $\phi_\pm(\cdot,z)$ are the Weyl solutions.
%

\subsection{Scattering}
Let us now add a compactly supported potential $f\in L^1_{loc}$, such that $f\neq0$. The goal of this paragraph is to prove that ``the perturbation affects the wavefunction'', except maybe for a discrete set of energies. We assume that $f$ satisfies the conditions for the elementary perturbation, which means that $\supp(f)\subset [-1/2,1/2]$ and that we are in one of the following cases:
\begin{enumerate}
 \item $u(x)=q_{am}(x)\sigma_1+q_{sc}(x)\sigma_3,$
 \item $u(x)=q_{el}(x)I_2$ and $\supp q_{el}\cap \supp (v_{per}^{am})^2+(v_{per}^{sc})^2$ has positive measure.
\end{enumerate}

For $\lambda$ in the support of the coupling constants $\lambda(\omega)$, we consider  the operator
\[H_\lambda=D_0+V_{per}+\lambda f\] on $L^2(\R,\C^2)$.
For $z\in S$, we define $u_+(\cdot,z)$ as the unique solution to 
\begin{equation}\label{Hpert}H_\lambda u_+(\cdot,z)=zu_+(\cdot,z)\end{equation} which coincides with $\phi_+$ on $(-\infty,-1/2)$. Since the two functions $\phi_\pm(\cdot,z)$ are a basis of the space of solutions to the unperturbed equation, we have that there exist some constants $a(z)$ and $b(z)$ such that, for $x\geq 1/2$, \begin{equation}\label{defab}u_+(x,z)=a(z)\phi_+(x,z)+b(z)\phi_-(x,z).\end{equation}

From the definition of the Floquet solutions and the fact that, for $E\in(E_-,E_+)$, $\rho_+(E)=\overline{\rho_-(E)}$, we get that $\phi_-(x,E)=\overline{\phi_+(x,E)}$. Thus, if we denote by  $u_-$ the solution to~\eqref{Hpert} which coincides with $\phi_-$ on $(-\infty,-1/2)$, we have that, for $x\geq 1/2$ and $E\in(E_-,E_+)$, 
 \begin{equation}\label{WeylFlo}u_-(x,E)=\overline{a(E)}\phi_-(x,E)+\overline{b(E)}\phi_-(x,E).\end{equation}
 
 The constancy of the Wronskian gives that
 \begin{equation}\label{Wronskian}
  |a(E)|^2-|b(E)|^2=1.
 \end{equation}

We have the following important result about the coefficients $a$ and $b$, the proof of which is similar to the one of Proposition~2.1 of \cite{DSS}.
\begin{prop}
 The functions $a(\cdot)$ and $b(\cdot)$, defined on $S$ as above, are branches of multi-valued analytic functions with at most algebraic singularities at boundaries of stability intervals.
\end{prop}

Our key lemma is the following; it states that, as soon as the perturbation is not identically 0, the reflection coefficient $b$ cannot be identically 0 on a stability interval.

\begin{lemma}\label{lem:b0}
 If $b(E)=0$ for all $E\in(E_-,E_+)$, then $\lambda=0$ is identically 0.
\end{lemma}
\begin{proof}
 If we assume that $b(E)=0$ for all $E\in(E_-,E_+)$, then, by analyticity, it is true for all $z\in S$. Then, for all $E\in(E_-,E_+)$ and $\eta>0$, $u_+(\cdot,E+\iu\eta)=a(E+\iu \eta)\phi_+(\cdot,E+\iu\eta)$ is exponentially decaying at $+\infty$: it is thus the Weyl solution for the perturbed equation. The Weyl-Titchmarsh function at point -1/2 for the half-line $(-1/2,+\infty)$  (see appendix) is then
 \[m_{V_{per}+\lambda f}(E+\iu\eta)=\frac{u_{+}^\downarrow(-1/2,E+\iu\eta)}{u_{+}^\uparrow(-1/2,E+\iu\eta)}=\frac{\phi_{+}^\downarrow(-1/2,E+\iu\eta)}{\phi_{+}^\uparrow(-1/2,E+\iu\eta)}=m_{V_{per}}(E+\iu\eta).\]
 Since both Weyl-Titchmarsh functions coincide in the upper part of $S$ and are both analytic on $\C^+$, they coincide on $\C^+$.
 Then, by appendix~\ref{app:A}, we get that $\lambda=0$
\end{proof}

As a consequence, since we have assumed that $f$ is not identically 0 and we know that $b$ is analytic, we have that the set $\{E\in(E_-,E_+):b(E)=0\}$ is discrete.

Let us now consider a gap: let $\alpha<E_-<E_+$ such that $(\alpha, E_-)$ is a maximal, non-trivial gap in the spectrum of $H_0$. Note that we always have $\alpha>-\infty$ since the perturbation is bounded (cf. \cite{BC}, remark after Proposition~5.2).  We consider a strip from which we have removed the  spectrum:
\[S'=\{z=E+\iu\eta, \alpha<E<E_+, \eta\in\R\}\backslash[E_-,E_+).\]
Let $\rho_i(z)$ ($i=1,2$) be the branches of~\eqref{char-eq}, such that $|\rho_1(z)|<1$ and $|\rho_2(z)|>1$ in $S'$. As before, these two functions are analytic in $S'$ with at worst algebraic singularities at $\alpha$, $E_-$ and $E_+$. Moreover, $\rho_1=\rho_+$ in the upper part of $S$ and $\rho_1=\rho_-$ in its lower part. For $z\in S'$ and $i=1,2$, we will denote by $v_i(z)$ the eigenvector of $g_0(z)$ corresponding to the eigenvalue $\rho_i(z)$; the $v_i$'s are analytic functions on $S'$. We denote by $\phi_i(\cdot,z)$ the solution to~\eqref{floquet} satisfying $\phi_i(-1/2,z)=v_i(z)$. We define in this case $u_i(\cdot, z)$ the solution to~\eqref{Hpert} which coincides  with $\phi_i(\cdot,z)$ on $(-\infty,-1/2)$. We have then that, for $x>1/2$, $u_i(x,z)=a_i(z)\phi_i(x,z)+b_i(z)\phi_j(x,z)$ ($i\neq j$), where, for each of both $i$'s, the two functions $a_i$ and $b_i$ are analytic in $S'$. Moreover, we see that, in the upper part (resp. lower part) of $S$, $a_1$ and $b_1$ (resp. $a_2$ and $b_2$) coincide with the functions $a$ and $b$ defined in the first case. Hence, they are not identically zero: the set \[\{E\in(\alpha,b):a_1(E)b_1(E)a_2(E)b_2(E)=0\}\] is discrete.

\subsection{Transfer matrices and Lyapunov exponent}
One of our main tool to get information on the nature of the spectrum of $H_\omega$ is the Lyapunov exponent. For $E\in\C$, $n\in\N$ and $\omega\in\Omega$, we denote by $g_E(n,\omega)$ the transfer matrix from $n-1/2$ to $n+1/2$ for $H_\omega$, i.e. the matrix such that, for all $u$ such that $H_\omega u=E u$, we have
\begin{equation}
 u(n+1/2)=g_E(n,\omega)u(n-1/2).
\end{equation}
We define to the transfer matrix from $1/2$ to $n+1/2$:
\[U_E(n,\omega)=g_E(n,\omega)...g_E(1,\omega).\]

We define the Lyapunov exponent by \begin{equation}\label{defLyap}
                                   \gamma(E)=\lim_{n\to\infty}\frac1n\ec(\|U_E(n,\omega)\|).
                                  \end{equation}
The existence of the limit comes from the subadditive ergodic theorem (see for example \cite[Theorem~IV.1.2]{CL}).

Our goal here is to prove that $\gamma(E)>0$ for all $E\in\R$, with the possible exception of a discrete set. To this purpose, we will use the theorem of Fürstenberg, of which we recall the statement.
We naturally define the action of $SL_2(\R)$ on the projective space $P(\R^2)$ from a quotient of the usual action on $\R^2$. 

\begin{theorem}[\cite{BL}, Theorem~II.4.1 and Proposition~II.4.3]\label{Furst}
 Let $(A_\omega^n)_{n\in\N}$ be a sequence of random variables with values in $SL_ 2(\R)$, with common distribution $\mu$. Let $G_\mu$ be the smallest closed subgroup of $SL_2(\R)$ containing the support of $\mu$. Let us assume that
 \begin{enumerate}
  \item $G_\mu$ is not compact.
  \item $\forall v \in P(\R^2)$, $\#\{gv: g\in G_\mu\}\geq 3$.
 \end{enumerate}
Then,
there exists a Lyapunov exponent $\gamma>0$ such that, almost surely
\begin{equation}
 \lim_{n\to\infty}\frac{1}{n}\log(\|A_\omega^n...A_\omega^1\|)=\gamma.
\end{equation}
\end{theorem}

We will apply this theorem to the group spanned by transfer matrices, which will be called \emph{Fürstenberg's group}. More precisely, we will take, for all $E\in\R$: \begin{equation}                                                                                            G_\mu(E)=\overline{\langle \{g_E^0(\omega):\omega\in\Omega\}\rangle}.                                                                                \end{equation}
We do not write the parameter $n$ since the matrices are identically distributed.
The constancy of the Wronskian implies that $G_\mu$ is included in $SL_2(\R)$. Note that the hypotheses of Theorem~\ref{Furst}  are preserved if we take a bigger group: we can assume that $\supp~\mu$ has only two elements. We can assume, without loss of generality, that these two elements are 0 and 1 since, if they are 2 real numbers $q_1$ and $q_2$, we can replace $V_{per}$ by $V_{per}+\sum_{n\in\Z}q_1u(\cdot-n)$ and $u$ by $(q_2-q_1)u$.

\begin{prop}\label{defM}
 There exists a discrete set $M\subset\R$ such that, for all $E\in\R\backslash M$, $G_\mu(E)$ satisfies the hypotheses of~\ref{Furst}. As a consequence, $\gamma(E)>0$ for all $E\in\R\backslash M$.
\end{prop}

We will treat separately the case where $E$ is in a stability interval of $H_0$ and the one where $E$ is in a spectral gap of $H_0$.

Let us start by the case where $E$ is in a stability interval. We have the following lemma.
\begin{lemma}
 Let $E$ be in a stability interval $(E_-,E_+)$ of $H_0$. Then, \begin{equation}
                                                                   g_1(E)=C(E)\tilde{g}_0(E)s(E)C(E)^{-1}
                                                                  \end{equation}
where
\begin{equation}
 C(E)=\begin{pmatrix}
             1&0\\\Re[c(E)]&\Im[c(E)]
            \end{pmatrix}
\text{ ; }\tilde{g}_0(E) =\begin{pmatrix}
            \Re[\rho(E)]&\Im[\rho(E)]\\-\Im[\rho(E)]&\Re[\rho(E)]
           \end{pmatrix}
\end{equation}
\begin{equation}
 s(E) =\begin{pmatrix}
            \Re[a(E)+b(E)]&\Im[a(E)+b(E)]\\-\Im[a(E)-b(E)]&\Re[a(E)-b(E)]
           \end{pmatrix}
\end{equation}
with $c(E)=c_+(E)=\overline{c_-(E)}$, $\rho(E)=\rho_+(E)=\overline{\rho_-(E)}$, $a(E)$ and $b(E)$ are respectively defined by \eqref{defc}, \eqref{defrho} and \eqref{defab}.
\end{lemma}
\begin{proof}
 The proof is similar to the one of Lemma~2.4 of \cite{DSS} and is a very simple matrix calculation. We omit the dependence in $E$ to alleviate notations. First, expressing the solutions $u_{N/D}$ in terms of the Weyl solutions $u_\pm$ gives
 \[g_1(E)=\begin{pmatrix}
                 u_+^\uparrow(1/2)& u_-^\uparrow(1/2)\\ u_+^\downarrow(1/2)& u_-^\downarrow(1/2)
                \end{pmatrix}
\begin{pmatrix}
 1&1\\c&\overline{c}
\end{pmatrix}^{-1}.\]
Then, by \eqref{defab} and \eqref{WeylFlo}, \[\begin{pmatrix}
                 u_+^\uparrow(1/2)& u_-^\uparrow(1/2)\\ u_+^\downarrow(1/2)& u_-^\downarrow(1/2)
                \end{pmatrix}=\begin{pmatrix}
                 \phi_+^\uparrow(1/2)& \phi_-^\uparrow(1/2)\\ \phi_+^\downarrow(1/2)& \phi_-^\downarrow(1/2)
                \end{pmatrix}
\begin{pmatrix}
 a&\overline{b}\\b&\overline{a}
\end{pmatrix}.\]
Finally, the definition of the Floquet solutions gives us that
\[\begin{pmatrix}
                 \phi_+^\uparrow(1/2)& \phi_-^\uparrow(1/2)\\ \phi_+^\downarrow(1/2)& \phi_-^\downarrow(1/2)
                \end{pmatrix}=\begin{pmatrix}
 1&1\\c&\overline{c}
\end{pmatrix}\begin{pmatrix}
\rho&0\\0&
\overline{\rho}\end{pmatrix}.\]
Putting all these equalities together, we get the result.
\end{proof}

The matrix $g_0$ corresponding to trivial scattering (i.e. the case $a=1$, $b=0$), we have 
\begin{equation}
   g_0(E)=C(E)\tilde{g}_0(E)C(E)^{-1}.   
\end{equation}

We will consider the group $\tilde{G}(E)$ defined as the subgroup of $SL_2(\R)$ generated by $\tilde{g}_0(E)$ and $s(E)$. It is easy to see that it is conjugate to $G(E)$, and thus that $G(E)$ satisfies the hypotheses of Theorem~\ref{Furst} if and only if $\tilde{G}(E)$ satisfies them.

Given a stability interval $(E_-,E_+)$, let $\tilde{M}_{(E_-,E_+)}=\{E\in(E_-,E_+):b(E)=0\}$ and choose $E\in(E_-,E_+)\backslash\tilde{M}_{(E_-,E_+)}$. We set $a(E)=A\e^{\iu\alpha}$ and $b(E)=B \e^{\iu\beta}$, with $A$ and $B $ real. We have thus
\begin{equation}
 s(E)=A\begin{pmatrix}
                    \cos \alpha&\sin \alpha\\-\sin\alpha&\cos \alpha
                   \end{pmatrix}
+B \begin{pmatrix}  \cos \beta&\sin \beta\\\sin\beta&-\cos \beta
                   \end{pmatrix}.
\end{equation}
In order to prove that $G(E)$ is not compact, we will exhibit a sequence of elements with an unbounded norm. An arbitrary element of $P(\R^2)$ can be represented by $v(\theta)=\begin{pmatrix}\cos \theta\\\sin \theta\end{pmatrix}$, for some $\theta\in[0,\pi)$. We find then that
\[s(E)v(\theta)=Av(\theta-\alpha)+B v(\beta-\theta).\]
In particular, if we take $\theta'=\frac12(\alpha+\beta)$, we have that
\[s(E)v(\theta')=(A+B )v((\beta-\alpha)/2).\]

Let us now consider the squared norm $R(\theta)=\|s(E)v(\theta)\|^2$. Using the fact that $A^2-B ^2=1$ (cf.~\eqref{Wronskian}), we find that 
\begin{equation}
 R(\theta)-1=2B (B +A\cos(\alpha+\beta-2\theta)).
\end{equation}
We have thus that $R(\theta)>1$ if and only if $\cos(\alpha+\beta-2\theta)>-\frac{B }{A}$, which happens on an interval of measure higher than $\frac\pi2$.
Thus, there exists a compact interval $K$ with measure larger than $\frac\pi2$ such that $\|s(E)v(\theta)\|>c>1$ for all $\theta$ in $K$. In particular, $s(E)$ applied to $v(\theta')$ produces a vector $v'$ with norm larger than one.

Concerning $g_0(E)$, as $E$ is in a stability interval, $|\rho(E)|=1$ and thus $g_0(E)$ is a rotation. The angle of this rotation can be taken in $[-\pi/2,\pi/2)$ and is not 0 since $\rho\neq\pm1$.

The sequence is constructed in the following way.
\begin{enumerate}
 \item Choose a $v(\theta_0)$ in $K$.
 \item Apply $s(E)$ to get a new vector $w$ with larger norm.
 \item If $\overline{w}$, the class of $w$ in $P(\R^2)$, is in $K$, then do (2) again.
 \item If $\overline{w}\notin K$, then apply $g_0(E)$ enough times to get a vector in $K$. It can be reached in a finite number of iterations since $g_0(E)$ has an angle in $[-\pi/2,\pi/2)$ and $K$ is an interval of length larger than $\pi/2$. Then, apply (2) again.
\end{enumerate}


To conclude the proof in the stability intervals, we will pick
\[M_{(E_-,E_+)}=\{E\in(E_-,E_+):E\in\tilde{M}_{(E_-,E_+)}\text{ or } D(E)=0\}.\]

We clearly have that $M_{(E_-,E_+)}$ is finite and, by the above, $\tilde{G}(E)$ is not compact when we take $E\in(E_-,E_+)\backslash M_{(E_-,E_+)}$. Moreover, when $D(E)\neq 0$, the matrix $\tilde{g}_ 0(E)$ is a rotation with an  angle different from $\frac\pi2$.
Therefore, powers of this matrix produce at least three distinct elements in the projective space.

Let us now consider the case where $E$ is in a spectral gap of $H_0$, which will be denoted $(E',E_-)$. Since in such an interval, we have $|D(E)|>2$,  $g_0(E)$ always have an eigenvalue with modulus larger than one. Hence, repeated iterations of $g_0(E)$ on the associated eigenvector will produce an unbounded sequence of elements of  $G(E)$, showing that this group is not compact.

To prove that condition (2) of Theorem~\ref{Furst} is satisfied in $(E',E_-)$ (outside a discrete interval), let us introduce the set \[M_{(E',E_-)}=\{E\in(E',E_-):a_1(E)b_1(E)a_2(E)b_2(E)=0\},                                            \]which is  discrete as we have seen in the last subsection, and take $E$ in $(E',E_-)\backslash M_{(E',E_-)}$.

Let us denote by $\bar{v}_1$ and $\bar{v}_2$ the directions of the eigenvectors of $g_0(E)$. Since we have chosen $E$ such that $a_i(E)$ and $b_i(E)\neq0$, we have that
\[s(E)\bar{v}_i\notin\{\bar{v}_1,\bar{v}_2\}\] for $i=1,2$.
Thus, if $\bar{v}\notin\{\bar{v}_1,\bar{v}_2\}$, then $\#\{g_0(E)^n\bar{v},n\in\Z\}=\infty$. If $\bar{v}$ is one the the eigenvectors, then an application of $s(E)$ followed by iterations of $g_0(E)$ gives an infinite orbit.
This proves the second condition for spectral gaps.

Finally, we conclude the proof by taking $M$ as the union of, first, the $M_{(E_-,E_+)}$ for all the stability intervals, second, all the $M_{(E',E_-)}$ for all the maximal spectral gaps, and, third, all the endpoints of the stability intervals.

\section{Hölder regularity of the Lyapunov exponent}\label{sec:4}
In this section, we will prove the Hölder regularity of the Lyapunov exponent on compact intervals. Fix $I$ a compact interval of $\R\backslash M$, where $M$ is the discrete set introduced in the previous section.
\begin{prop}\label{reguLyap}
 The Lyapunov exponent defined by~\eqref{defLyap} is uniformly Hölder continuous on $I$, i.e. there exist $\alpha>0$ and $C$ such that, for all $E$, $E'\in I$,
 \[|\gamma(E)-\gamma(E')|\leq C|E-E'|^\alpha.\]
\end{prop}

The proof will use the following theorem, whose statement has been written (in a more general setting) by Boumaza.

\begin{theorem}[\cite{Boumazabis}, Theorem~1]\label{thm:Boum}
 Let $(A_\omega^n(E))$ a sequence of i.i.d. random matrices of $SL_2(\R)$ depending on a real parameter $E$. Let $\mu_E$ be the common distribution of the $(A_\omega^n(E))$, $G_{\mu_E}$ the smallest closed subgroup of $SL_2(\R)$ containing the support of $\mu_E$ and $\gamma(E)$ the associated Lyapunov exponent. We fix a compact interval $I\subset\R$ and we assume that for all $E\in I$ we have:
 \begin{enumerate}
  \item $G_{\mu_E}$ is non-compact and strongly irreducible.
  \item There exist $C_1>0$, $C_2>0$, independent of $n$, $\omega$ and $E$ such that
  \begin{equation}\label{Gronwall1}
   \|A_\omega^n(E)\|^2\leq \exp(C_1+|E|+1)\leq C_2.
  \end{equation}
\item There exists $C_3>0$, independent of $n$, $\omega$ and $E$ such that for any $E$, $E'\in I$:
  \begin{equation}\label{Gronwall2}
   \|A_\omega^n(E)-A_\omega^n(E')\|\leq  C_3|E-E'|.
   \end{equation}
 \end{enumerate}
Then, there exist two real numbers $\alpha>0$ and $C>0$ such that
\begin{equation}
 \forall E, E'\in I; \quad |\gamma(E)-\gamma(E')|\leq C|E-E'|^\alpha.
\end{equation}
\end{theorem}

We apply this theorem for $A_\omega^n(E)=g_\omega^n(E)$, the transfer matrices defined in the previous section.
This will directly give Proposition~\ref{reguLyap}; we only need to check the hypotheses of the theorem.
The first one has already been proven in the previous section.

Let us now prove~\eqref{Gronwall1}.
We have the following lemma.
\begin{lemma}\label{lem:Gronwall1}
 Let $\psi=\begin{pmatrix}\psi_\uparrow\\\psi_\downarrow\end{pmatrix}$ be a solution to $D_0\psi+V\psi=0$, with $V=q_{am}\sigma_1+q_{sc}\sigma_3+q_{el}I_2\in L^1_{loc}(\R,\mathcal{M}_2(\R))$. Then, for all $x$, $y\in\R$,
 \begin{equation} \label{majo1}
  |\psi_\uparrow(x)|^2+|\psi_\downarrow(x)|^2\leq  (|\psi_\uparrow(y)|^2+|\psi_\downarrow(y)|^2)\exp\left(2\int_{\min(x,y)}^{\max(x,y)}|q_{am}(t)|+|q_{sc}(t)|\d t\right).
 \end{equation}
\end{lemma}
\begin{proof}
 If we write in coordinates the equation satisfied by $\psi$, we find:
 \[-\psi_\downarrow'+(q_{sc}+q_{el})\psi_\uparrow+q_{am}\psi_\downarrow=0\] and 
 \[\psi_\uparrow'+q_{am}\psi_\uparrow+(q_{el}-q_{sc})\psi_\downarrow=0.\]
 
 We have then, denoting $R(t)=|\psi_\uparrow(t)|^2+|\psi_\downarrow(t)|^2$, for almost every $t\in\R$,
  \begin{align*}R'(t)=&2\Re (\psi_\uparrow'\overline{\psi_\uparrow}+\psi_\downarrow'\overline{\psi_\downarrow})\\
  =&2\Re(-q_{am}|\psi_\uparrow|^2-(q_{el}-q_{sc})\psi_\downarrow\overline{\psi_\uparrow}+(q_{sc}+q_{el})\psi_\uparrow\overline{\psi_\downarrow}+q_{am}|\psi_\downarrow|^2)\\
  =&2\Re(q_{am}(|\psi_\downarrow|^2-|\psi_\uparrow|^2)+2q_{sc}\psi_\uparrow\overline{\psi_\downarrow}).
  \end{align*}
Hence, we have $|R'(t)|\leq 2(|q_{am}(t)|+|q_{sc}(t)|)|R(t)|$ and we conclude with Grönwall's lemma.
\end{proof}
We immediately get~\eqref{Gronwall1}, taking $V=V_{per}+V_\omega-E I_2$; we even have something independent of $E$.

Let us now turn to the proof of~\eqref{Gronwall2}. We use the following lemma.
\begin{lemma}\label{lem:Gronwall2}
Let $\psi_1=\begin{pmatrix}\psi_{1,\uparrow}\\\psi_{1,\downarrow}\end{pmatrix}$ and $\psi_2=\begin{pmatrix}\psi_{2,\uparrow}\\\psi_{2,\downarrow}\end{pmatrix}$ be solutions to $D_0\psi_i+V_i\psi_i=0$, where, for $i=1,2$, $V_i$ is a $L_{loc}^1$ function with values in 2-by-2 matrices, and such that, for some $y$, $\psi_1(y)=\psi_2(y)$. Then, we have, for all $x$ in $\R$,
\begin{equation}
 |\psi_1(x)-\psi_2(x)|\leq|\psi_2(y)| \exp\left(\int_{\min(y,x)}^{\max(y,x)}|V_1(t)|+|V_2(t)|\d t\right)\times\int_{\min(y,x)}^{\max(y,x)}|V_1(s)-V_2(s)|\d s.
\end{equation}

\end{lemma}
\begin{proof}
 We have \[D_0(\psi_1-\psi_2)=V_1\psi_1-V_2\psi_2=V_1(\psi_1-\psi_2)+(V_1-V_2)\psi_2.\]
 Thus, 
 \[(\psi_1-\psi_2)'=J^{-1}\left[V_1(\psi_1-\psi_2)+(V_1-V_2)\psi_2\right]\] and then, assuming without loss of generality that $y\leq x$,
 \[(\psi_1-\psi_2)(x)=J^{-1}\left[\int_y^xV_1(t)(\psi_1-\psi_2)(t)\d t+\int_y^x(V_1(t)-V_2(t))\psi_2(t)\d t\right].\]
 We get then that 
 \[|(\psi_1-\psi_2)(x)|\leq\left[\int_y^x|V_1(t)||\psi_1(t)-\psi_2(t)|\d t+\int_y^x|V_1(t)-V_2(t)||\psi_2(t)|\d t\right].\]
 Grönwall's lemma gives us that
 \[|(\psi_1-\psi_2)(x)|\leq\int_y^x|V_1(t)-V_2(t)||\psi_2(t)|\d t\exp \left(\int_y^x|V_1(s)|\d s\right).\]
 Finally, the previous lemma gives us that for all $t\in[y,x]$, \[|\psi_2(t)|\leq |\psi_2(y)|\exp\left(\int_y^t|V_2(s)|\d s\right),\] which concludes the proof.
\end{proof}
We directly get~\eqref{Gronwall2} by taking $V_1=V_{per}+V_\omega-E$ and $V_2=V_{per}+V_\omega-E'$. We have thus satisfied all the hypotheses of Theorem~\ref{thm:Boum}, which gives us Proposition~\ref{reguLyap}.

\section{Kotani theory and regularity of the integrated density of states}\label{sec:5}
We prove here Theorem~\ref{thm:reguDOS}, which is Hölder regularity of the density of states; we get this regularity from the one we got in the previous section for the Lyapunov exponent.

We recalled in the introduction that, in~\cite{Zdos}, we had defined the density of states measure by~\eqref{def-dos}.
Nevertheless, with this expression, it would be difficult to recover a Wegner estimate. As a consequence, we give here another definition. We restrict the operator on an interval with what we will call \emph{Dirichlet boundary conditions}. Note that they are not the usual Dirichlet conditions, where $\psi=0$ on the boundary, but we choose this name by analogy with the second-order case where we consider the vector $\begin{pmatrix}
 \psi\\ \psi'                                                                                                                                                                                                                                                                                                                                                                                                                                
                                                                                                                                                                                                                                            \end{pmatrix}$.

Given $x\in\R$ and $L>0$, we define the operator $H_{\omega,x,L}$ , called the \emph{restricted operator with Dirichlet boundary condition} to the interval $I=(x-L/2,x+L/2)$, as the operator acting as $H_\omega$ on the domain
\[\Dom(H_{\omega,x,L})=\{\psi=\begin{pmatrix}
                          \psi^\uparrow\\\psi^\downarrow
                         \end{pmatrix}\in H^1(I)\text{ such that } \psi^\uparrow(x-L/2)=\psi^\uparrow(x+L/2)=0\}.\] We will use the notation $H_{\omega,L}$ for $H_{\omega,0,L}$.

                         We define the density of states in the following way.
\begin{defi}
 For all bounded measurable $\phi$, the density of states is 
 \[\nu(\phi)=\lim_{L\to\infty}\frac1L\ec(\tr(\phi(H_{\omega,L}))).\]
\end{defi}
This definition is not the same as the one we have used in~\cite{Zdos}, but we prove in appendix~\ref{app:B} that both are equivalent
.
It is easy to see that $\nu$ is a positive linear functional on the bounded continuous functions and, thus comes from a measure on $\R$ which will be denoted by $\nu$ too.
Since our Hamiltonian is not bounded from below, we cannot define the integrated density of states in the usual way.
Instead of that, we take
\begin{equation}
 N(E)=\left\{\begin{array}{l}
        -\nu((E,0])\text{ if }E<0\\
        \nu((0,E])\text{ if } E\geq 0 
        \end{array}
 \right.
\end{equation}

We introduce too $N_0$, which is the integrated density of states for the free Dirac operator. An explicit calculation proves that the eigenvalues of this operator with Dirichlet boundary conditions on $[-L/2,L/2]$ are the $(k\pi/L)_{k\in\Z}$, and thus $N_0(E)=\frac{E}{\pi}$.


%


In order to get between the Lyapunov exponent and the integrated density of states, we prove the following result, called Thouless formula.
\begin{prop}
 Let $N$ and $N_0$ be the integrated densities of states, respectively of the operator $H_\omega$ and $D_0$, and $\gamma$ be the Lyapunov exponent for the operator $H_\omega$. Then, there exists $\alpha\in\R$ such that, for every $E\in\R$,
 \begin{equation}
  \gamma(E)=-\alpha+\int_\R\log\left|\frac{E-t}{t-\iu}\right|\d(N-N_0)(t).
 \end{equation}
\end{prop}
The original formula for  Schrödinger operators, in the discrete case, is given by $\gamma(E)=-\alpha+\int_\R\log\left|{E-t}\right|\d N(t)$\cite[Proposition~VI.4.3]{CL}. Nevertheless, it is easy to see that the integral is not convergent in our case. Several methods have been used in the case of continuous Schrödinger operators: one of them is to introduce the normalization term $t-\iu$~\cite{DSS}, another one is to consider the differences $N-N_0$ and $\gamma-\gamma_0$ \cite{AS}. None of these two methods provides a convergent integral in our case, but a combination of both works. Note that, since the free Dirac operator has purely absolutely continuous spectrum on the whole real line, its Lyapunov exponent is identically zero (cf.~\cite[Corollary~VII.3.4]{CL}). That is why, in most of the cases, we do not write $\gamma_0$ in the formulas.
In all the following, the splitting between $V_{per}$ and $V_\omega$ will not be used. Hence, we define $W_\omega=V_{per}+V_\omega$.
We have the following lemma.
\begin{lemma}
 There exists a constant $C$, depending only on $\|W_\omega\|_\infty$, such that for all $E\in\R$, we have that
 \begin{equation}
  |N(E)-N_0(E)|\leq C.
 \end{equation}
 \end{lemma}
\begin{proof}
 We will give a proof that uses the Hellmann-Feynman formula. Indeed, the restricted operator has discrete spectrum. We will give the proof for $E>0$, it is similar for $E<0$. For $\theta\in[0,1]$, let $H_{\omega,L,\theta}=D_0+\theta W_\omega$ on $[-L/2,L/2]$ with Dirichlet boundary conditions. Then, there exist continuous functions $E_k(\cdot)$, defined on $[0,1]$, such that the spectrum of $H_{\omega,L,\theta}$ is the set $\{E_k(\theta),k\in\Z\}$ for all $\theta$. Moreover, we have that each of these functions is differentiable and that, denoting by $\psi_k(\theta)$ a \emph{normalized} eigenfunction associated with $E_k(\theta)$, \[\frac{\d E_k}{\d \theta}=\langle \psi_k(\theta), W_\omega\psi_k(\theta)\rangle.\] Thus, the modulus of the derivative is smaller than $\|W_\omega\|_\infty$. Then, the number of eigenvalues of $H_{\omega,L,1}$ in $[0,E]$ cannot be larger than the number of eigenvalues of $H_{\omega,L,0}$ in $[-\|W_\omega\|_\infty, E+\|W_\omega\|_\infty]$, and thus
 \[N(E)\leq N_0(E+\|W_\omega\|_\infty)-N_0(-\|W_\omega\|_\infty)=\frac{1}{\pi}(E+2\|W\|_\infty).\]
 Symmetrically, we have that for all $E'$, \[N_0(E')\leq N(E'+\|W_\omega\|_\infty)-N(-\|W_\omega\|_\infty)\] and thus
 \[N(E)\geq N_0(E-\|W_\omega\|_\infty)+N(-\|W_\omega\|_\infty).\]
\end{proof}

In the following, we need to have an $\R$-ergodic operator. In \cite{Kirsch}, Kirsch proves that, given a $\Z$-ergodic operator $H_\omega$, we can construct an $\R$-ergodic operator $\overline{H}_{\overline{\omega}}$, on a wider probability space, with, for each $\overline{\omega}$,  $\overline{H}_{\overline{\omega}}$ is unitarily equivalent to $H_\omega$ for some $\omega$.

Since we apply Kirsch's suspension procedure, the expectation must here be understood as both expectation on $\Omega$ and average value on $[0,1]$ for the potentials. It will be denoted by $\eb$.

The next step towards the Thouless formula is to introduce a tool named Kotani $w$-function. We use it in the way it has been introduced by Sadel and Schulz-Baldes for Dirac operators in \cite{SSB}.

For the operator we have introduced in the beginning, we denote by $m_+$ the Weyl-Titchmarsh function on $(0,+\infty)$, as it is defined in Appendix~\ref{app:A}. If we can decompose $W_\omega=W_{\omega,am}\sigma_1+W_{\omega,sc}\sigma_3+W_{\omega,el}I_1$, we define the function
\[w(z)=-\eb(W_{\omega,am}+m_+(z)(W_{\omega,el}-W_{\omega,sc}-z)).\]

Similarly, we introduce the corresponding function for the free Dirac operator $D_0$:
\[w_{0}(z)=-zm_{0}(z)\]
where $m_{0}(z)=\iu$  is the Weyl function associated with the free operator.

This function will provide us a link between the Lyapunov exponent and the integrated density of states through the Green's function. We recall that the Green's function $G(z,\cdot,\cdot)$ is the integral kernel of the resolvent $(H_\omega-z) ^{-1}$.  We start by  the following theorem of Sadel and Schulz-Baldes, which links the Green's function and the Lyapunov exponent. Here $G(z)=G(z,x,x)$, and we drop the dependence in $x$ since we take the expectation $\eb$.
\begin{theorem}[\cite{SSB}, Theorem~5]\label{thm:SSB}
 Let $\Im(z)\neq0$. Then,
 \begin{enumerate}
  \item $\gamma(z)=-\Re(w(z))$.
  \item $\partial_zw(z)=\eb(\tr G(z))$.
 \end{enumerate}
\end{theorem}
These equalities are still valid when the potential is zero, i.e. for $w_0$, $G_0$ and $\gamma_0$.

To get the link with the integrated density of states, we prove the following lemma:
\begin{lemma}\label{lem:G-IDS}
 Let $z\in\C\backslash\R$. Then,
 \[\eb(\tr(G(z)))=\int_\R\frac{\d\nu(E)}{E-z}.\]
\end{lemma}
\begin{proof}
 We follow the proof given in \cite{polyHB} for Schrödinger operators. For $l>0$, we denote by $\chi_l$ the characteristic function of $[-l/2,l/2]$.
 $\R$-ergodicity implies that, 
  for all $l>0$ and all bounded borelian set $B$, \[\nu(B)=\eb(\frac1l\tr(\chi_lE_\omega(B)\chi_l)).\]
 By taking the limit $l\to0$, we get that
 \[\nu(B)=\eb(\tr(\delta_0E_\omega(B)\delta_0)).\]
 The operator in the trace is well defined on $H^1(\R,\C^2)$ since this space is included in $\mathcal{C}(\R,\C^2)$.
 Thus,
 \begin{align*}
  \int_\R\frac{\d\nu(E)}{E-z}&=\int_\R\frac{1}{E-z}\d \eb(\tr(\delta_0E_\omega\delta_0))(E)\\
  &=\eb(\tr(\delta_0\int_\R\frac{1}{E-z}\d E_\omega(E)\delta_0))\\
  &=\eb(\tr(\delta_0(H_\omega-z)^{-1}\delta_0))\\
  &=\eb(\tr (G(z))).
 \end{align*}
\end{proof}

The second point of Theorem~\ref{thm:SSB} and Lemma~\ref{lem:G-IDS} can be combined to get  a link between the imaginary part of the $w$ functions and the integrated densities of states.
\begin{lemma}
 There exists $c\in\R$ such that, for all $E\in\R$, we have
 \begin{equation}
  \lim_{a\to0^+} \Im\left(w(E+\iu a)-w_0(E+\iu a)\right)=\pi(N(E)-N_0(E))+c.
 \end{equation}
\end{lemma}

\begin{proof}
 According to the second point of Theorem~\ref{thm:SSB}, we know that, for all $z\in\C$ with $\Im(z)>0$, $w'(z)=\eb(G(z))$. With the previous lemma, it gives that
 \[w'(z)-w_0'(z)=\int_\R\frac{\d (\nu-\nu_0)(E')}{E'-z}.\]
 With an integration by parts, in which we use the fact that $(N(E)-N_0(E))/(E-z)$ tends to 0 in $\pm\infty$ since the numerator is bounded, we get that
 \[w'(z)-w_0'(z)=\int_\R\frac{N(E')-N_0(E')}{(E'-z)^2}\d E'.\]
 
 If we integrate, we find that
 \begin{equation}\label{eq:w}w(z)-w_0(z)=c+\int_\R\frac{1+E'z}{(E'-z)(1+E'^2)}(N(E')-N_0(E'))\d E'.\end{equation}
 
 Taking the imaginary part, for $z=E+\iu a$, we have that
 \begin{align*}\Im(w(E+\iu a)-w_0(E+\iu a))&=c'+a\int_\R\frac{N(E')-N_0(E')}{(E-E')^2+a^2}\d E'\\
 &=c'+\int_\R\frac{N(E+au)-N_0(E+au)}{1+u^2}\d u.
 \end{align*}
In the last line, we have done the change of variable $u=(E'-E)/a$.

Since $N$ and $N_0$ are right-continuous (as they are cumulative distribution functions), we get
\[\lim_{a\to 0+}\Im(w(E+\iu a)-w_0(E+\iu a))=(N(E)-N_0(E))\int_\R\frac{1}{1+u^2}\d u =\pi(N(E)-N_0(E)).\]
\end{proof}

We can now conclude the proof of the Thouless formula using the first point of Theorem~\ref{thm:SSB}.
We know that, for all $z$ with $\Im(z)\neq 0$, $\gamma(z)=\gamma(z)-\gamma_0(z)=-\Re(w(z)-w_0(z)).$
Thus, we get from~\eqref{eq:w} that 
\[\gamma(z)=-c-\Re\left(\int_\R\frac{1+E'z}{(E'-z)(1+E'^2)}(N(E')-N_0(E'))\d E'\right).\]

If we integrate by parts, we find that 
\[\gamma(z)=-c+\Re\left(\int_\R\log\left|\frac{t-z}{t-\iu}\right|\d(\nu-\nu_0)(t)\right).\]
We can get rid of the terms with $\lim_{t\to\pm\infty}\log\left|\frac{t-z}{t-\iu}\right|(N(t)-N_0(t))$ since the first factor goes to 0 in $\pm\infty$ and the second one is bounded. We see that the integral is real.

It remains to take the limit when $z$ becomes real. For the right-hand side, the pointwise convergence of the integrand is clear. If $z=\iu y$ (for positive $y$),  then the convergence as $y$  goes to 0 is monotone. In the other cases, we use the fact that $\log\left|\frac{t-E-\iu a}{t-\iu}\right|=\log\left|\frac{t-E-\iu a}{t-E-\iu}\right|+\log\left|\frac{t-E-\iu}{t-\iu}\right|$ to conclude.

The limit in the left-hand side is proven as in Proposition~8.2.14 of \cite{polyHB}.

With the Thouless formula, we can prove Hölder regularity of the integrated density of states with the same arguments as~\cite{DSS}. We give the proof here for the sake of completeness.

The tool is the Hilbert transformation defined, for a square-integrable function $\psi$, by
\[(T\psi)(x)=\frac1\pi\lim_{\epsilon\to0^+}\int_{|x-t|>\epsilon}\frac{\psi(t)}{x-t}\d t.\]

We use the two following results on this transformation, the proof of which being in \cite{CL}:
\begin{enumerate}
 \item For all $\psi\in L^2(\R)$, $T\psi$ is a square-integrable function and $T^2\psi=-\psi$ almost everywhere with respect to Lebesgue measure.
 \item If a function $\psi$ is Hölder-continuous on some interval $[x_0-a,x_0+a]$, then $T\psi$ is Hölder-continuous on $[x_0-\frac a2,x_0+\frac a2]$.
\end{enumerate}

We are now ready to prove Theorem~\ref{thm:reguDOS}.

\begin{proof}[Proof of Theorem~\ref{thm:reguDOS}]
 The first point is that, since $t\mapsto \log|(E-t)/(t-\iu)|$ is $\d (N-N_0)$-integrable, then
 \[\lim_{\epsilon\to0^+}\int_{E-\epsilon}^{E+\epsilon}|\log\left|\frac{t-E-\iu a}{t-\iu}\right||\d (N-N_0)(t)=0,\]
 and thus \begin{equation}
           \lim_{\epsilon\to0^+}|\log(\epsilon)|(N-N_0)(E+\epsilon)-(N-N_0)(E-\epsilon)=0.
          \end{equation}
Thus, $(N-N_0)$ is log-Hölder continuous and in particular continuous. Since $N_0$ is continuous, we have the same property for $N$.

Let $E_0\in I$ and $a>0$ such that $J=[E_0-4a,E_0+4a]\subset\R\backslash M$, which implies that $\gamma$ is Hölder-continuous in $J$ by Proposition~\ref{reguLyap}. Taking $\psi(t)=(N-N_0)(t)\chi_J(t)$, we note that $T^2\psi(t)=-(N-N_0)(t)$ for every $t$ in $J$. 

For $|E-E_0|<4a$, we can compute the Hilbert transform of $\psi$, using the Thouless formula to make $\gamma$ appear and the previously proven Hölder continuity to get rid of the limit in $\epsilon$:
\begin{eqnarray*}
 \pi(T\psi)(E)&=&\gamma(E)+\alpha-\int_{[t-E_0|>4a}\log\left|\frac{E-t}{t-\iu}\right|\d (N-N_0)(t)\\&+&\frac12\int_{E_0-4a}^{E_0+4a}\log(t^2+1)\d(N-N_0)(t)\\&-&\log(E_0-E+4a)(N-N_0)(E_0+4a)+\log(E-E_0+4a)(N-N_0)(E_0-4a).
\end{eqnarray*}
We see then that $T\psi$ is Hölder-continuous in $[E_0-2a,E_0+2a]$, which implies that $T^2\psi$ is Hölder-continuous in $[E_0-a,E_0+a]$. But we know that $T^2\psi=N-N_0$ almost everywhere in this interval and thus, since both are continuous, they are equal everywhere in the interval. We get thus that $N-N_0$ is Hölder-continuous in $[E_0-a,E_0+a]$ and $N$ has this property too since $N_0$ is Lipschitz.

By a compactness argument, we get that the Hölder continuity is uniform in any compact interval.
\end{proof}

\section{Multiscale analysis and localization}\label{sec:6}
\subsection{The requirements of multiscale analysis}
This section will be devoted to the proof of spectral and dynamical localization for our model. We use a method called \emph{bootstrap multiscale analysis}, which has been developed by Germinet and Klein in \cite{GKbootstrap}. Even if we do not need to write a proof, we have to state the hypotheses.

The first one is the possibility to use generalized eigenfunctions.

Let $\Hilb=L^2(\R,\d x;\C^2)$. Given $\nu>1/4$, we put, for $x\in\R$, $\langle x\rangle=\sqrt{1+x^2}$ and we define the weighted spaces $\Hilb_\pm$ as
 $$
 \Hilb_\pm=L^2(\R,\langle x\rangle ^{\pm4\nu}\d x;\C^2).
 $$

The sesquilinear form
 $$
 \langle\phi_1,\phi_2\rangle_{\Hilb_+,\Hilb_-}=\int\bar{\phi_1}(x)\cdot\phi_2(x)\d x
 $$
where $\phi_1\in\Hilb_+$ and $\phi_2\in\Hilb_-$ is the duality map.

We set $T$ to be the self-adjoint operator on $\Hilb$ given by multiplication by the function $\langle x\rangle ^{2\nu}$ ; note that $T^{-1}$ is bounded.

\begin{property}[SGEE]
 We say that an ergodic random operator $H_\omega$ satisfies the strong property of  generalized eigenfunction expansion (SGEE) in some open interval $I$ if, for some $\nu>1/4$, 
 \begin{enumerate}
  \item The set
 $$
 \Dom^\omega_+=\{\phi\in\Dom(H_\omega)\cap \mathcal{H}_+; 
 H_\omega\phi\in \Hilb_+\}
 $$
 is dense in $\Hilb_+$ and is an operator core for $H_\omega$ with probability one.
 
 \item There exists a bounded, continuous function $f$ on $\R$, strictly positive on the spectrum of $H_\omega$
 such that
 \begin{equation*}
  \ec\left\{[ \tr(T^{-1}f(H_\omega) \Pi_\omega(I)T^{-1})]^2\right\}<\infty,
 \end{equation*}
$\Pi_\omega$ being the spectral projection associated with $H_\omega$.
 \end{enumerate}
\end{property}
%
%
%
\begin{defi}
A measurable function $\psi:\R\to\C^2$ is said to be a generalized eigenfunction of $H_\omega$ with generalized eigenvalue $E$ if $\psi\in\Hilb_-\backslash\{0\}$ 
and 
 $$
 \langle H_\omega\phi,\psi \rangle_{\Hilb_+,\Hilb_-}=E\langle \phi,
 \psi \rangle_{\Hilb_+,\Hilb_-},\ \mbox{ for all } \phi\in\Dom_+^\omega.
 $$
\end{defi}
%
%
%
As explained in \cite{klein}, when (SGEE) holds, a generalized eigenfunction which is in $\Hilb$ is a bona fide eigenfunction. 
Moreover, if $\mu_\omega$ is the spectral measure for the restriction of $H_\omega$ to the Hilbert space $E_\omega(I)\Hilb$, then $\mu_\omega$-almost every $E$ is a generalized eigenvalue of $H_\omega$.

The following properties are about finite volume operators
%
%
%
%
Given an interval $\Lambda_L(x)=[x-L/2,x+L/2]$, we recall that the localized operator $H_{\omega,x,L}$ is the operator $H_\omega$ restricted to $\Lambda_L(x)$ with Dirichlet boundary condition, as defined in the previous section.

We can then define $R_{\omega,x,L}(z)=(H_{\omega,x,L}-z)^{-1}$ the resolvent of 
$H_{\omega,x,L}$ and $\Pi_{\omega,x,L}(\cdot)$ its spectral projection.

\begin{defi}
We say that  an ergodic random family of operators $H_\omega$ is standard \cite{klein} if for each $x\in\Z$, $L\in\N$ there is a measurable map $H_{\cdot,x,L}$ from $\Omega$ to self-adjoint operators
 on $L^2(\R,\C^2)$ such that 
 $$
   U(y)H_{\omega,x,L}U(-y)=H_{\tau_y\omega,x+y,L}
 $$ 
 where $\tau$ and $U$ define the ergodicity:
 $$
 U(y)H_\omega U(y)^*=H_{\tau_y(\omega)}.
 $$
\end{defi}
It is easy to see that the restriction of the operators to intervals with Dirichlet boundary conditions  makes $H_\omega$ a  standard operator.

We now enumerate the properties which are needed for multiscale analysis to be performed, yielding thus various localization properties.

\begin{defi}
An event is said to be based in a box $\Lambda_L(x)$ if it is determined by 
conditions on the finite volume operators 
$(H_{\omega,x,L})_{\omega\in\Omega}$.
\end{defi}
\begin{property}[IAD]
Events based in disjoint boxes are independent.
\end{property}

The following properties are to hold in a fixed open interval $I$. We will denote by $\chi_{x,L}$ the characteristic function of $\Lambda_L(x)$ and 
$\chi_x := \chi_{x,1}$. We also denote  by $\Gamma_{x,L}$ the characteristic function of the union of two regions near the boundary of $\Lambda_L(x)$: $[x-(L-1)/2,x-(L-3)/2]\cup[x+(L-3)/2,x+(L-1)/2]$. 
\begin{property}[SLI]\label{SLI}
For any compact interval $J\subset I$ there exists a finite constant $\kappa_J$ such that, given $L$, $l'$, $l''\in 2\N$, $x$, $y$, $y'\in\Z$ with 
 $\Lambda_{l''}(y)\subset\Lambda_{l'-3}(y')\subset\Lambda_{L-6}(x)$, then, for $\mathbb{P}$-almost every $\omega$, if $E\in J$ with $E\notin\sigma(H_{\omega,x,L })\cup\sigma(H_{\omega,y',l'})$
 we have
 \begin{equation}\label{eq:SLI-cond}
  \|\Gamma_{x,L}R_{\omega,x,L}(E)\chi_{y,l''}\|\leqslant\kappa_J\|\Gamma_{y',l'}R_{\omega,y',l'}(E)\chi_{y,l''}\|\|\Gamma_{x,L}R_{\omega,x,L}(E)\Gamma_{y',l'}\|.
 \end{equation}

\end{property}

\begin{property}[EDI]
 For any compact interval $J\subset I$ there exists a finite constant $\tilde{\kappa}_J$ such that for $\mathbb{P}$-almost every $\omega$, given a generalized eigenfunction $\psi$ of $H_\omega$
 with generalized eigenvalue $E\in J$, we have, for any $x\in\Z$ and $L\in2\N$ with $E\notin\sigma(H_{\omega,x,L})$, that
 \begin{equation*}
  \|\chi_x\psi\|\leqslant\tilde{\kappa}_J\|\Gamma_{x,L}
  R_{\omega,x,L}(E)\chi_x\|\|\Gamma_{x,L}\psi\|.
 \end{equation*}
\end{property}

\begin{property}[NE]
 For any compact interval $J\subset I$ there exists a finite constant $C_J$ such that, for all $x\in\Z$ and $L\in2\N$,
 \begin{equation*}
  \ec \left(\tr \left(\Pi_{\omega,x,L}(J)\right)\right)\leqslant C_J L.
 \end{equation*}

\end{property}

\begin{property}[W]\label{propW}
 For some $b\geqslant 1$, there exists for each compact subinterval $J$ of $I$ a constant $Q_J$ such that
 \begin{equation}\label{eq:wegner}
  \mathbb{P}\{\dist (\sigma(H_{\omega,x,L}),E)<\eta\}\leqslant Q_J\eta L^{b} ,
\end{equation}
for any $E\in J$, $0<\eta<\frac{1}{2}\dist(E_0,\sigma(H_0))$, $x\in\Z$ and $L\in2\N$.
\end{property}

The last property depends on several parameters: $\theta>0$, $E_0\in\R$ and $L_0\in6\N$.
\begin{property}[H1($\theta$, $E_0$, $L_0$)]
 \begin{equation*}
  \pc\left\{\left\|\Gamma_{0,L_0}R_{\omega,0,L_0}(E_0)
  \chi_{0,L_0/3}\right\|
  \leqslant\frac{1}{L_0^\theta}\right\}>1-\frac{1}{841}.
 \end{equation*}
\end{property}

 These properties are the hypotheses of the bootstrap multiscale analysis.
\begin{defi}
Given $E\in\R$, $x\in\Z$ and $L\in6\N$ with $E\notin \sigma(H_{\omega,x,L})$, we say that the box $\Lambda_L(x)$ is $(\omega,m,E)$-regular for a given $m>0$ if
\begin{equation}\label{eq:regular}
  \left\|\Gamma_{x,L}R_{\omega,x,L}(E)\chi_{x,L/3} \right\|
  \leqslant \e^{-mL/2}.
\end{equation}
\end{defi}
%
%
%
In the following, we denote
 $$
 [L]_{6\N}=\sup\{n\in6\N | n\leqslant L\}.
 $$
\begin{defi}
For $x$, $y\in\Z$, $L\in6\N$, $m>0$ and $I\subset\R$ an interval, we denote
\begin{equation*}
\begin{split}
 & R(m,L,I,x,y) \\
 & =\left\{\omega;\text{for every }E'\in I\text{ either }\Lambda_L(x)\text{ or }\Lambda_L(y)\text{ is }(\omega,m,E')\text{-regular.}\right\}.
\end{split}
\end{equation*}
The multiscale analysis region $\Sigma_{MSA}$ for $H_\omega$ is the set of $E\in\Sigma$  for which there exists some open interval $I \ni E$
such that, given any $\zeta$, $0<\zeta<1$ and $\alpha$, $1<\alpha<\zeta^{-1}$, there is a length scale $L_0\in6\N$ and a mass $m>0$ so if we set $L_{k+1}=[L_k^\alpha]_{6\N}$, $k=0,1, \ldots $,
we have
\begin{equation*}
 \mathbb{P}\left\{R(m,L_k,I,x,y)\right\}\geqslant 1-e^{-L_k^\zeta}
\end{equation*}
for all $k\in\N$, $x$, $y\in\Z^d$ with $|x-y|>L_k$.
\end{defi}
%
%
%
\begin{theorem}[Multiscale analysis - Theorem 5.4 p136 of \cite{klein}]\label{pi}
 Let $H_\omega$ be a standard ergodic random operator with (IAD) and properties (SLI), (NE) and (W) fulfilled in an open interval $I$. For $\Sigma$ being the almost sure spectrum of $H_\omega$ and for $b$ as in \eqref{eq:wegner}, given $\theta>b$, for each $E\in I$ there exists a finite scale
 $\mathcal{L}_\theta(E)=\mathcal{L}_\theta(E,b)>0$ bounded on compact subintervals of $I$ such that, if for a  given $E_0\in\Sigma\cap I$ we have (H1)($\theta$, $E_0$, $L_0$) at some scale $L_0\in6\N$ with $L_0>\mathcal{L}_\theta(E_0)$,
 then $E_0\in\Sigma_{MSA}$.
\end{theorem}
%
%
%
\begin{theorem}[Localization - Theorem~6.1 p139 of \cite{klein}]\label{thm:main-section4}
 Let $H_\omega$ be a standard ergodic operator with (IAD) and properties (SGEE) and (EDI) in an open interval $I$. Then,
 $$
 \Sigma_{MSA}\cap I\subset\Sigma_{EL}\cap\Sigma_{SSEHSKD}\cap I.
 $$
\end{theorem}
%
%

Thus, to prove Theorems~\ref{thm:loc-spec} and~\ref{thm:loc-dyna}, we only have to prove that all these hypotheses are satisfied. (SGEE), (SLI), (IAD) and (EDI) have already been proven in \cite[Proof of Theorem~4.1]{BCZ}. Similarly, (NE) is proven in the same paper, in the proof of Theorem~4.2. Even if there are extra hypotheses in that paper, they are not used in the proof of these specific assumptions.

It remains thus to proof the Wegner estimate (W) and the initial length-scale estimate (H1), which will be done in the last two subsections.

\subsection{Proof of the Wegner estimate}
We now prove the Wegner estimate, which is one the main ingredients which will enable us to do the multiscale analysis. We will get it from the Hölder continuity of the integrated density of states, as it is done in \cite{DSS}. 

We consider a compact interval $I=[a,b]\subset\R\backslash M$, where $M$ is the discrete set introduced in Proposition~\ref{defM}. In order to get uniform estimates on $I$, it will at some point be convenient to consider a larger set $I_\xi=[a-\xi, b+\xi]$, for some $\xi>0$ such that $I_\xi\subset \R\backslash M$.

We will prove the following:
\begin{prop}\label{propWeg}
 For every $\beta\in(0,1)$ and every $\sigma>0$, there exists $L_0\in\N$ and $\alpha>0$ such that
 \begin{equation}
  \pc(\dist(E,\sigma(H_{\omega,L}))\leq \e^{-\sigma L^\beta})\leq\e^{-\alpha L^\beta}
 \end{equation}
for all $E\in I$ and $L\geq L_0$.
\end{prop}
This Wegner estimate is not exactly similar to the one stated in Proposition~\ref{propW}, but it is possible to use this version for multiscale analysis, as explained in~\cite{CKM} and~\cite[Remark~4.6]{klein}.

Our first lemma is the following, which is Lemma~5.2 of \cite{DSS}. Since the proof does not depend on the precise shape of the operators but only on properties of transfer matrices, we will not give it here.
\begin{lemma}[\cite{DSS}, Lemma~5.2]\label{lem:WegLyap}
 There exists $\alpha>0$, $\delta>0$, $n_0\in\N$ such that for all $E\in I$, $n\geq n_0$ and $x$ with norm 1, we have
 \begin{equation}
  \ec\{\|U_E(n)x\|^{-\delta}\}\leq \e^{-\alpha n},
 \end{equation}
where we recall that $U_E(n)=g_E(n)...g_E(1)$, where $g_E(k)$ is the transfer matrix from $k-1/2$ to $k+1/2$ (we have omitted the dependence in $\omega$).
\end{lemma}

The proof of Proposition~\ref{propWeg} is based on a decomposition in several events. 
To bound one of them, we use the following lemma.

\begin{lemma}\label{lemWeg}
 There exists $\rho>0$ and $C<\infty$ such that, for every $E\in I$ and every $\epsilon>0$, we have for $L\geq L_0$
 \begin{equation}\label{proba1}\begin{split}
 & \pc\{\text{There exists }E'\in(E-\epsilon,E+\epsilon)\text{ and }\phi\in\Dom(H_{\omega,L}),\|\phi\|=1,\text{ such that }\\& (H_{\omega,L}-E')\phi=0, |\phi_\downarrow(-L/2)|^2+|\phi_\downarrow(L/2)|^2\leq\epsilon^2\}\leq CL\epsilon^\rho.\end{split}
 \end{equation}
\end{lemma}
\begin{proof}
 The proof is basically similar to the one of Damanik, Sims and Stolz in \cite{DSS}, which follows itself the arguments of \cite{CKM} for the discrete case.
 
Let us fix $E\in I$, $\epsilon>0$ and $L$. For $k\in\Z$, we denote by $\Lambda_k$ the interval $[kL-L/2,kL+L/2]$ and by $H_{\omega,k}$ the operator $H_\omega$ restricted to $\Lambda_k$ with Dirichlet boundary conditions.

Let $A_k$ be the event \begin{equation}\begin{split}\label{event}\{H_{\omega,k}\text{ has an eigenvalue }E_k\in(E-\epsilon,E+\epsilon)\text{ such that the corresponding normalized eigenfunction }\phi_k\\\text{ satisfies }|\psi_{k,\downarrow}(kL-L/2)|^2+|\psi_{k,\downarrow}(kL+L/2)|^2\leq\epsilon^2\}\end{split}.\end{equation} The fact that the $V_\omega$ are i.i.d. implies that $p=\pc(A_k)$ is independent of $k$ and equals in particular the left-hand side of \eqref{proba1}.

For $n\in\N$, we define the interval $\Lambda^{(n)}=\cup_{k=-n}^n\Lambda_k=[-nL-L/2,nL+L/2]$ and the operator $H_\omega^{(n)}$ as the restriction of $H_\omega$ to $\Lambda^{(n)}$ with Dirichlet boundary conditions. For a given $\omega$, we assume that the event $A_k$ occurs for a number $j$ of distinct values of $k$ which will be denoted $k_1$,...,$k_j$.
For each $1\leq i\leq j$, we denotes by $\phi_i$ an eigenfunction of $H_{\omega,k_i}$ satisfying the conditions given in~\eqref{event}.

Our first step is to construct, for each $k_i$, a function $\tilde{\phi}_i$  in the domain of $H_\omega^{(n)}$ which will have support in $\Lambda_{k_i}$, and will be an approximate eigenfunction in the sense that there exists $C_2>0$ such that, for every $i$, 
\[\|(H_\omega^{(n)}-E)\tilde{\phi}_i\|\leq C_2\epsilon.\]

To this purpose, we introduce, as in \cite{DSS}, a smooth function $\chi$ such that $0\leq\chi\leq1$, $\chi(x)=0$ for $x\leq0$, and $\chi(x)=1$ for $x\geq1$. For a given $i$, we denote by $x_i^\pm=k_iL\pm L/2$. We define the function $\hat{\phi}_i$ on $\Lambda^{(n)}$ by
\[\hat{\phi}_i(x)=\left\{\begin{array}{c c l} 0&\text{if}&x\notin\Lambda_{k_i}\\
                          \chi(x-x_i^-)\phi_i(x)&\text{if}&x_i^-\leq x\leq x_i^-+1\\
                          \phi_i(x)&\text{if}&x_i^-+1\leq x\leq x_i^+-1\\
                          \chi(x_ i^+-x)\phi_i(x)&\text{if}&x_i^+-1\leq x\leq x_i^+
                         \end{array}\right.
\]

It is clear that $\hat{\phi}_i$ is in the domain of $H_\omega^{(n)}$ and that its norm is less than $\|\phi_i\|=1$. We know that $\phi_i$ is an eigenfunction of $H_{\omega,k_i}$ corresponding to an eigenvalue $E_{k_i}\in(E-\epsilon,E+\epsilon)$. We can thus decompose
$\|(H_\omega^{(n)}-E)\hat{\phi}_i\|\leq\|(H_\omega^{(n)}-E_{k_i})\hat{\phi}_i\|+\epsilon$. Let us now estimate the first term. To this purpose, we compute the pointwise value of $(H_\omega^{(n)}-E)\hat{\phi}_i$. First, for $x$ outside $\Lambda_{k_i}$, $\hat{\phi}_i(x)=0$, and thus the expression is 0.
Second, if $x\in(x_i^-+1,x_ i^+-1)$, $\hat{\phi}_i(x)=\phi_i(x)$ so, by definition of $\phi_i$, $(H_\omega^{(n)}-E_{k_i})\hat{\phi_i}$ is 0 almost everywhere in this interval. Finally, if $x\in(x_i^-,x_ i^-+1)$ (the case $x\in(x_i^+-1,x_ i^+)$ is totally similar),
\[(H_\omega-E_{k_i})\hat{\phi}_i(x)=(D_0+W_\omega-E_{k_i})\hat{\phi}_i(x)=\chi'(x-x_i^-)J\phi_i(x),\] where all the other terms cancel since $\phi_i$ is an eigenfunction of $H_{\omega,k_i}$.

Hence,
\begin{align*}
 \|(H_\omega-E_{k_i})\hat{\phi}_i\|^2&=\int_{x_i^-}^{x_i^-+1}| \chi'(x-x_i^-)\phi_i(x)|^2\d x+ \int_{x_i^+-1}^{x_i^+}| \chi'(x_i^+-x)\phi_i(x)|^2\d x\\
 &\leq \|\phi_i\|_{L^\infty(x_i^-,x_i^-+1)}^2\int_{x_i^-}^{x_i^-+1}| \chi'(x-x_i^-)|^2\d x+\|\phi_i\|_{L^\infty(x_i^+-1,x_i^+)}^2\int_{x_i^+-1}^{x_i^+}| \chi'(x_i^+-x)|^2\d x\\
 &\leq C_3(|\phi_i(x_ i^+)|^2+|\phi_i(x_ i^-)|^2)\\
 &\leq 2 C_2 \epsilon^2
\end{align*}
where the constant $C_3$ depends only on the single site potential, the single site distribution and the function $\chi$. We used Lemma~$\ref{lem:Gronwall1}$ in the third step. We define then $\tilde{\phi}_i=\hat{\phi_i}/\|\hat{\phi_i}\|$. The definition of $\hat{\phi}$ implies that its norm is greater than 1/2 if $\epsilon$ in  small enough. Thus, $\tilde{\phi}$ satisfies all the properties given in the left-hand side of~\eqref{proba1}.

Note that, because of the disjointness of supports, we have that if $i\neq j$,
\[\langle\tilde{\phi}_i,H_{\omega}^{(n)},\tilde{\phi}_j\rangle=0=\langle\tilde{\phi}_i,H_{\omega}^{(n)},H_{\omega}^{(n)}\tilde{\phi}_j\rangle.\]
We know then, by \cite[Lemma~A.3.2]{ST}, that $H_\omega^{(n)}$ has at least $j$ eigenvalues (counted with multiplicity) in the interval $[E-C_2\epsilon, E+C_2\epsilon]$.

On the other hand, ergodicity implies that, for almost every $\omega$,
\[p=\lim_{n\to\infty}\frac{1}{2n+1}\#\{l\in\{-n,...,n\}\text{ such that }A_l\text{ occurs}\}.\] If we take $\epsilon$ small enough so that $[E-C_2\epsilon, E+C_2\epsilon]\subset I_\xi$, we have that, for almost every $\omega$,
\begin{align*}
 p&\leq\lim_{n\to\infty}\frac{1}{2n+1}\#\{\text{eigenvalues of }H_\omega^{(n)}\text{ in }[E-C_2\epsilon, E+C_2\epsilon]\}\\
 &=L\lim_{n\to\infty}\frac{1}{(2n+1)L}\#\{\text{eigenvalues of }H_\omega^{(n)}\text{ in }[E-C_2\epsilon, E+C_2\epsilon]\}\\
 &=L(N(E+C_2\epsilon)-N(E-C_2\epsilon)).
\end{align*}
  We know, from Theorem~\ref{thm:reguDOS}, that there are constants $\rho=\rho(I_\xi)>0$ and $C_1=C_1(I_\xi)<\infty$ such that, for every $E, E'\in I_\xi$,
  \begin{equation}
|N(E)-N(E')|\leq C_1 |E-E'|^\rho.
  \end{equation}
  We get then that $p\leq C_1L\epsilon^\rho$.
\end{proof}

\begin{proof}[Proof of Proposition~\ref{propWeg}]
 Let $\beta$, $\sigma$ and $I$ be as above. For $L\in\N$ odd, we define $n_L=\lfloor\tau(L/2)^\beta\rfloor+1$, for some $\tau>0$ which will be chosen later. For $E\in I$ and $\theta>0$, we define the following events:
 \[A_\theta^{(E,L)}=\{\|g_E(-(L+1)/2+n_L)\cdot...\cdot g_E(-(L-1)/2)\begin{pmatrix}                                                                                          0\\1                                                                                         \end{pmatrix}\|>\e^{\theta(L/2)^\beta}\},\]
\[B_\theta^{(E,L)}=\{\|g_E((L+1)/2-n_L)^{-1}\cdot...\cdot g_E((L-1)/2)^{-1}\begin{pmatrix}                                                                                          0\\1                                                                                         \end{pmatrix}\|>\e^{\theta(L/2)^\beta}\},\]
\[C^{(E,L)}=\{\dist(E,\sigma(H_{\omega,L}))\leq\e^{-\sigma L^\beta}\}.\]
Our goal is to estimate $\pc(C^{(E,L)})$. We will decompose it in a sum of 4 terms, depending or whether the events $A_\theta^{(E',L)}$ and $B_\theta^{(E',L)}$ are satisfied in an interval around $E$.

First, we will bound \[p_1=\pc\left(C^{(E,L)}\cap \bigcap_{|E-E'|\leq\e^{-\sigma L^\beta}}(A_{\kappa/2}^{(E',L)}\cap B_{\kappa/2}^{(E',L)})\right).\]
We see that, if we are in this event, then, putting $\epsilon=\max(\e^{-\sigma L^\beta},\e^{-\frac12\kappa(L/2)^\beta})$, $H_{\omega,L}$ has an eigenvalue in $[E-\epsilon,E+\epsilon]$ with an associated eigenfunction satisfying the conditions of Lemma~\ref{lemWeg}: we have then that
\[p_1\leq CL\max(\e^{-\sigma L^\beta\rho},\e^{-\frac12\kappa(L/2)^\beta\rho}).\]

The second term to bound will be
\[p_2=\pc\left(A_\kappa^{(E,L)}\cap B_\kappa^{(E,L)}\cap \bigcup_{|E-E'|\leq\e^{-\sigma L^\beta}}(A_{\kappa/2}^{(E',L)})^c\right).\]
We can prove that this probability is smaller than $\e^{-\alpha_2(L/2)^\beta}$ for some $\alpha_2>0$ if $\tau$ is small enough, using \eqref{Gronwall2}.

We have the same bound for
\[p_3=\pc\left(A_\kappa^{(E,L)}\cap B_\kappa^{(E,L)}\cap \bigcup_{|E-E'|\leq\e^{-\sigma L^\beta}}(B_{\kappa/2}^{(E',L)})^c\right).\]

Finally, Lemma~\ref{lem:WegLyap} directly gives us that
\[p_4=\pc((A_\kappa^{(E,L)})^c)+\pc((B_\kappa^{(E,L)})^c)\leq2\e^{-\frac12\tau\alpha_1(L/2)^\beta}\] for $L$ large enough.

Seeing that 
$\pc(C^{(E,L)})\leq p_1+p_2+p_3+p_4$, we get the result.
\end{proof}

\subsection{Initial length-scale estimate}
In this section, we prove an initial length-scale estimate of the type (H1)($\theta$, $E$, $L_0$). We begin by giving a collection of  large deviation results on the convergence $\frac1n\log\|U_E(n)\|\to\gamma(E)$, which can be taken without modification from \cite{DSS}. In all the following, we fix $E\in\R\backslash M$.

We recall that \[U_E(n)=g_E(n,\omega)...g_E(1,\omega)\] and that the Lyapunov exponent is 
\[\gamma(E)=\lim_{n\to\infty}\frac{1}{n}\ec(\|U_E(n,\omega)\|).\]

\begin{lemma}[\cite{DSS}, Lemma~6.2]
 There exists $\alpha>0$ such that for every $\epsilon>0$ and $x\neq 0$, one has
 \[\limsup_{n\to\infty}\frac1n\log\pc(|\log\|U_E(n)x\|-n\gamma(E)|>n\epsilon)<\alpha.\]
\end{lemma}

In particular, we get that for every $\epsilon>0$ and $x\neq0$, there exists $n_0\in\N$ such that for all $n\geq n_0$
\begin{equation}\label{asymp-U}
 \pc\left(\e^{(\gamma(E)-\epsilon)n}\leq\|U_E(n)x\|\leq\e^{(\gamma(E)+\epsilon)n}\right)\geq1-\e^{-\alpha n}.
\end{equation}

\begin{lemma}[\cite{DSS}, Lemma~6.3]
 Fix $y$ with $\|y\|=1$. For all $\epsilon>0$, there are $n_0\in\N$ and $\delta_0>0$ such that
 \[\sup_{x\neq0}\pc\left(\frac{|\langle U_E(n)x,y\rangle|}{\| U_E(n)x\|}<\e^{-\epsilon n}\right)<\e^{-\delta_0n}\text{ for }n\geq n_0.\]
 \end{lemma}
 
 With these two lemmas, we can prove that, with large probability, the matrix elements $|\langle U_E(n)x,y\rangle|$ grow exponentially at a rate almost equal to the Lyapunov exponent.
 \begin{coro}[\cite{DSS}, Corollary~6.4]\label{coro:H1}
  Let $\|x\|=\|y\|=1$. Then, for every $\epsilon>0$, there exists $\delta>0$ and $n_0\in\N$ such that 
  \[\pc\left(|\langle U_E(n)x,y\rangle|\geq \e^{(\gamma(E)-\epsilon) n}\right)\geq1-\e^{-\delta n}\text{ for }n\geq n_0.\]
 \end{coro}
We can now prove the following theorem.
\begin{theorem}\label{thm:H1}
 For every $\epsilon>0$, there exists $\delta>0$ and $L_0\in\N$ such that for $L\geq L_0$, $L\in3\N\backslash 6\N$, we have
 \begin{equation}
  \pc(\Lambda_L(0)\text{ is }(\omega,\gamma(E)-\epsilon,E)-\text{regular})\geq1-\e^{-\delta L}.
 \end{equation}

\end{theorem}
\begin{proof}
We denote by $u_\pm$ the solutions to $H_\omega u=E u$ satisfying the Dirichlet boundary condition at $\pm L/2$, in the sense that $u_\pm(\pm L/2)=(0,1)^t$. We will denote by \[W(u_+,u_-)=u_+^\uparrow u_-^\downarrow-u_-^\uparrow u_+^\downarrow\] the (constant) Wronskian of $u_+$ and $u_-$.
This Wronskian is 0 if and only if $E\in\sigma(H_{\omega,\Lambda})$.

 The first fact we have to notice is that the integral kernel of the resolvent $G_\Lambda(E,x,y)$ is given by
 \begin{equation}\label{defG}
  G_\Lambda(E,x,y)=\frac{1}{W(u_+,u_-)}\left\{\begin{array}{r l }u_+(x)\otimes u_-(y)&\text{ for }x\geq y\\
                            u_-(x)\otimes u_+(y)&\text{ for }x< y
                            \end{array}\right.
 \end{equation}

 Indeed, if we take, for some $\phi\in L^2(\Lambda,\C^2)$, $\psi(x)=\int_\Lambda G_\Lambda(E,x,y)\phi(y)\d y$, we see first that $\psi$ is in the domain of $H_\omega$ (i.e. it is in $H^1$ and satisfies the boundary conditions). An explicit calculation of $H_\omega\psi$ completes the proof.

Thus, to prove that the operator norm of $\Gamma_{0,L} R_{\omega,\Lambda}\chi_{0,L/3}$ is small enough (with high probability), we will use Schur's test, which implies to control its integral kernel.
We begin by proving that $|W(u_+,u_-)|$ is big with probability large enough.
To this purpose, we remark that, in particular, if we take the functions at $L/2$,
\[W(u_+,u_-)=\langle u_-(L/2),(-1,0)^t\rangle=\langle U_E(L/2,-L/2)(0,1)^t,(-1,0)^t\rangle.\]

Because of stationarity, we get from Corollary~\ref{coro:H1} that, for $L\geq L_0$,
\[\pc\left(|W(u_+,u_-)|\geq \e^{(\gamma(E)-\epsilon)L}\right)\geq 1-\e^{-\delta L}.\] 

Let $x\in[(L-3)/2,(L-1)/2]$ and $y\in[-L/6,L/6]$. By Lemma~\ref{lem:Gronwall1}, we get that there exists $C>0$, independent of $L$ and $\omega$, such that for all such $x$
\[|u_+(x)|\leq C.\]

On the other hand, for all such $y$,
\[|u_-(y)|=|U_E(y,-L/2)(0,1)^t|\leq \|U_E(y,\lfloor y+1/2\rfloor-1/2)\||U_E(\lfloor y+1/2\rfloor-1/2,-L/2)(0,1)^t|.\]

Lemma~\ref{lem:Gronwall1} again gives that there exists $C>0$ independent of $L$ such that, for all $\omega\in\Omega$, $E\in\R$ and $y\in\R$, \[\|U_E(y,\lfloor y+1/2\rfloor-1/2)\|\leq C\] and, by stationarity and \eqref{asymp-U}, there exists $\alpha>0$ such that
\begin{equation}
 \pc\left(|A_E(\lfloor y+1/2\rfloor-1/2,-L/2)(0,1)^t|\leq \e^{(\gamma(E)+\epsilon)2L/3}\right)\geq1-\e^{-\alpha L/3}\text{ if } L/3\geq L_0.
\end{equation}

Using the expression of $G_\Lambda(E,x,y)$ given by~\eqref{defG}, all the previous estimates give that
\begin{equation}
 \pc\left(|G_\Lambda(E,x,y)|\leq C\e^{-(\gamma(E)-5\epsilon)L/3}\right)\geq 1-\e^{-\delta L}-\e^{-\alpha L/3}
\end{equation}
if $L$ is large enough. We obviously have the same estimate for $x\in[-(L-1)/2,-(L-3)/2]$ and $y\in[-L/6,L/6]$.

Thus, we get that there exists some $\delta>0$ such that, with probability at least $1-\e^{-\delta L}$, 
\[\sup_x\int|\Gamma(x)G_\Lambda(E,x,y)\chi_{0,L/3}(y)|\d y\leq \e^{-(\gamma(E)-\epsilon)L/3}\]
and
\[\sup_y\int|\Gamma(x)G_\Lambda(E,x,y)\chi_{0,L/3}(y)|\d y\leq \e^{-(\gamma(E)-\epsilon)L/3},\]
which gives the result through Schur's test.
\end{proof}

It remains to prove that the value of $L_0$ is uniform on any compact interval. In fact, we are even able to prove the stronger following result.
\begin{coro}
 For every $E\in I$ and every $\beta\in(0,1)$, $\sigma$, $\epsilon>0$, let $\alpha>0$, $\delta>0$ and $L_0\in\N$ be as Theorem~\ref{thm:H1} and Proposition~\ref{propWeg}. For every $\epsilon'>\epsilon$ and every $L\geq L_0$, we define
 \[\zeta_L=\frac12\e^{-2\sigma L_\beta}\left(\e^{-(\gamma(E)-\epsilon')L/3}-\e^{-(\gamma(E)-\epsilon)L/3}\right).\]
 Then, we have for every $L\geq L_0$:
 \begin{equation}
  \pc\left(\forall E'\in(E-\zeta_L,E+\zeta_L),\ \Lambda_L(0)\text{ is }(\omega, \gamma(E)-\epsilon',E')-\text{regular}\right)\geq1-\e^{-\delta L}-\e^{-\alpha L^\beta}.
 \end{equation}

\end{coro}
\begin{proof}
 We know that, with probability higher than $1-\e^{-\delta L}-\e^{-\alpha L^\beta}$, we have both that $\Lambda_L(0)$ is $(\omega,\gamma(E)-\epsilon,E)$-regular and that $\dist(E,\sigma(H_{\omega,L}))\leq \e^{-\sigma L^\beta}$. In this case, for every $E'\in (E-\zeta_L,E+\zeta_L)$, $E'\notin\sigma(H_{\omega,L})$. We can thus write the resolvent equation:
 \begin{eqnarray*}\|\Gamma_{0,L}R_{\omega,0,L}(E')\chi_{0,L/3}\|&=
  \|\Gamma_{0,L}(R_{\omega,0,L}(E)+(E-E')R_{\omega,0,L}(E')R_{\omega,0,L}(E))\chi_{0,L/3}\|\\
  &\leq\|\Gamma_{0,L}R_{\omega,0,L}(E)\chi_{0,L/3}\|+|E-E'|\|R_{\omega,0,L}(E')R_{\omega,0,L}(E)\|\\
  &\leq\e^{-(\gamma(E)-\epsilon)L/3}+|E-E'|\e^{2\sigma L^\beta}\\
  &\leq \e^{-(\gamma(E)-\epsilon')L/3}.
 \end{eqnarray*}
We conclude the proof with the fact that $|E-E'|\leq \zeta_L$.
\end{proof}

\appendix
\section{Some results on Weyl-Titchmarsh theory of Dirac operators}\label{app:A}
This section is devoted to the proof of Lemma~\ref{lem:b0}. The fundamental tool is the Weyl-Titchmarsh function.
Given a Dirac operator with potential $D=D_0+B$, a real number $x_0$ and a complex number $z$ with $z\notin\R$, we define the Weyl-Titchmarsh number $m_B(z,x_0)$ as the unique complex number such that the solution to $Du=zu$ on $(x_0,+\infty)$ satisfying $u(x_0)=\begin{pmatrix} 1\\m_B(z,x_0)\end{pmatrix}$ is in $L^2$ in a neighbourhood of $+\infty$. Existence and uniqueness come from the ``limit-point property'' of Dirac operators (cf.~\cite{Weidmann} and~\cite{CG}).

 Weyl-Titchmarsh theory for Dirac operators is studied in \cite{CG}. In this paper, they consider operators of the form
\[D_0+V_{am}\sigma_1+V_{sc}\sigma_3.\] Such operators will be called in \emph{normal form}.
They prove the following theorem of local uniqueness.

\begin{theorem}[\cite{CG},~Theorem~1.2]\label{loc-unic}
 Fix $x_0\in \R$ and suppose that there are two potentials $B_1$ and $B_2$ such that, for all $j$, $B_j\in L^1([x_0,x_0+R])$ for all $R>0$ and it is in normal form on $(x_0,+\infty)$. Denote by $m_{B_j}$ the Weyl-Titchmarsh functions at $x_0$ for the half-line $(x_0,+\infty)$ associated with the operator with potential $B_j$. 
 If, for all $\epsilon>0$, \begin{equation}\label{eg-WTF}
                            |m_{B_1}(z)-m_{B_2}(z)|=O(e^{-2\Im(z)(a-\epsilon)})
                           \end{equation}
as $z$ goes to infinity along a ray $\rho_1$ with argument in $(0,\pi/2)$ and along a ray $\rho_2$ with argument in $(\pi/2,\pi)$, then \[B_1(x)=B_2(x)\text{ for almost every }x\in[x_0,x_0+a].\]
\end{theorem} 

According to the proof of Lemma~\ref{lem:b0}, we will take in all the following $x_0=-1/2$.
Now, let us consider the case where the operator is not necessarily in normal form -- under each of our two hypotheses. It is possible to prove (cf.~\cite{BEKT}) that the operator with an electrostatic potential $D=D_0+V_{am}\sigma_1+V_{sc}\sigma_3+V_{el}I_2$ is unitarily equivalent to an operator in normal form using the gauge transformation \[R_{2\phi}=\begin{pmatrix}                                                                                                        \cos(2\phi)&\sin(2\phi)\\-\sin(2\phi)&\cos(2\phi)                                                                                                                      \end{pmatrix}\text{ where }\phi(x)=\int_{-1/2}^xV_{el}(t)\d t;\]
we have thus $D=R_{2\phi}^{-1}\tilde{D}R_{2\phi}$, with
\[\tilde{D}=D_0+\tilde{V}_{am}\sigma_1+\tilde{V}_{sc}\sigma_3,\]
where \[\begin{pmatrix}\tilde{V}_{am}\\\tilde{V}_{sc}\end{pmatrix}=R_{2\phi}\begin{pmatrix} V_{am}\\V_{sc}                                                                           \end{pmatrix}.\]

We see that if $Du=zu$ for some $z\in\C$, then $\tilde{D}R_{2\phi}u=zR_{2\phi}u$. Moreover, for all $x\in\R$, $|u(x)|=|(R_{2_\phi}u)(x)|$. In particular, if $u$ is square-integrable near $+\infty$, $R_{2\phi}u$ has the same property.
These two facts imply that, if $u(\cdot,z)$ is the Weyl solution associated with $D$ at the point $-1/2$, then $R_{2\phi}(\cdot)u(\cdot,z)$ is the Weyl solution associated with $\tilde{D}$ at the same point. But, since we have chosen the transformation $R_{2\phi}$ such that $R_{2\phi}(-1/2)$ is the identity matrix, we have that the Weyl-Titchmarsh functions of the two operators at this point are equal. 

Now, let us consider the family of the potentials $V_\lambda=V_{per}+\lambda f$, where $V_{per}$ satisfies the hypotheses for the periodic potential, $f$ satisfies one of the conditions for the elementary random potential and $\lambda$ is in a subset $\Xi$ of $\R$. We want to prove that, if  Weyl-Titchmarsh functions  of the Dirac operator $D_{V_\lambda}$ coincides with the one of $D_{V_{per}}$ on $\C^+$, then $\lambda=0$. Let us assume that we are in this case for some $\lambda$. 
Then, the associated operators in normal form $\tilde{D}_{V_1}$ and $\tilde{D}_{V_2}$ satisfy the same property and, by the local uniqueness theorem,
\begin{equation}\label{eq:WT}R_{2\phi_0}\begin{pmatrix} v_{per}^{am}\\v_{per}^{sc}\end{pmatrix}=R_{2\phi_\lambda}\begin{pmatrix} V_{am,\lambda}\\V_{sc,\lambda}\end{pmatrix},\end{equation} where, $\phi_\lambda$ is an antiderivative of $V_{el,\lambda}$ ($v_{per}^{el}$ if $\lambda=0$) with $\phi_\lambda(-1/2)=0$. 

In our case (1), we have to deal with the case where $V_{el,\lambda}=v_{per}^{el}$ for all $\lambda\in\Xi$, which implies $\phi_\lambda=\phi_0$. As a consequence, in this case, \eqref{eq:WT} implies that $V_{am,\lambda}=v_{per}^{am}$ and $V_{sc,\lambda}=v_{per}^{sc}$, and thus $\lambda=0$.

 Case (2) corresponds to $\begin{pmatrix} V_{am,\lambda}\\V_{sc,\lambda}\end{pmatrix}=\begin{pmatrix} v_{per}^{am}\\v_{per}^{sc}\end{pmatrix}$. Thus, for all $x$ in the support of this vector, we have $\phi_\lambda(x)=\phi_0(x)+k\pi$, for some $\pi\in\Z$ and, by taking the derivative, $v_{per}^{el}(x)\sigma_3+\lambda f(x)=v_{per}^{el}(x)\sigma_3$. Since we have assumed that there exists $x\in\supp\begin{pmatrix} v_{per}^{am}\\v_{per}^{sc}\end{pmatrix}$ such that $f(x)\neq 0$, then $\lambda=0$.
\section{Density of states with Dirichlet boundary conditions}\label{app:B}
In section~\ref{sec:5}, we define the density of states as a linear form on bounded measurable functions, using operators restricted to boxes with Dirichlet boundary conditions. In~\cite{Zdos}, we had defined the density of states by\[\tilde{\nu}(\phi)=\lim_{L\to\infty}\frac{1}{L}\tr\left(\chi_L\phi(H_\omega)\chi_L\right),\]
where $\chi_L$ is the characteristic function of $[-L/2,L/2]$.

We recall that we have defined the operator $H_{\omega,L}$ as the restriction of $H_\omega$ to $(-L/2,L/2)$ with Dirichlet boundary conditions.
We have the following result, which gives both the convergence and the equivalenc of the two definitions.
\begin{prop}
 For all $\phi$, we have almost surely 
 \[\lim_{L\to\infty}\frac1L\tr(\phi(H_{\omega,L}))=\tilde{\nu}(\phi).\]
\end{prop}

%
The proof of this proposition is very similar to the one of Section~4 \cite{Zdos}, where we use periodic boundary conditions. Henceforth, we will not give it entirely, but we give only the following lemma, which corresponds to Lemma~3 of~\cite{Zdos}. We consider here operators restricted to $(0,L)$ since the calculations are easier in this case, but the result is obviously the same for operators restricted to $(-L/2,L/2)$.
\begin{lemma}
\label{BS}
  Let $L>0$ and $D_{0,L}^{(D)}$ be the Dirac operator restricted to $(0,L)$ with Dirichlet boundary condition. Let $H=f(x)g(D_{0,L}^{(D)})$ on $L^2(0,L)$ with $f\in L^2(0,L)$ and $\sum_{n\in\Z}|g(\frac{n\pi}{L})|^2<\infty$ . Then, $H$ is Hilbert-Schmidt and 
  \begin{equation}
   \|H\|_2\leq \frac{C}{\sqrt{L}}\|f\|_2\left(\sum_{n\in\Z}|g(\frac{n\pi}{L})|^2\right)^{1/2}.
  \end{equation}
\end{lemma}

\begin{proof}
 The first step is to diagonalize the operator $D_{0,L}^{(D)}$. With a standard calculation, we find that the eigenvalues are the $E_k=\frac{k\pi}{L}$ for $k\in\Z$ and the associated (normalized) eigenvectors are the $\frac{1}{\sqrt{L}}\begin{pmatrix}
 \sin(E_k x)\\
 \cos(E_k x) \end{pmatrix}$. As a consequence, the operator $H$ has a matrix-valued integral kernel given by
 \[K(x,t)=\frac1L\begin{pmatrix}
    f(x)\sum_{k\in\Z}g(E_k)\sin(E_kt)\sin(E_kx)&0\\0& f(x)\sum_{k\in\Z}g(E_k)\cos(E_kt)\cos(E_kx)
   \end{pmatrix}\]
   
   According to Theorem~2.11 of \cite{simon}, the operator $H$ is Hilbert-Schmidt if and only if its integral kernel is $L^2$ and the Hilbert-Schmidt norm of the operator is equal to the $L^2$ norm of the integral kernel.
   By standard trigonometric manipulations, we are left to bound the $L^2$ norm of \[f(x)\sum_{k\in\Z}g(E_k)\cos(E_k(t\pm x)).\]
This norm is clearly equal to
\[\|f\|_2\|\sum_{k\in\Z}g(E_k)\cos(E_k(\cdot))\|_2.\]
Evaluating the second factor gives the result.
\end{proof}
\bibliographystyle{plainurl}
\bibliography{biblio}
\end{document}